\documentclass[%
 reprint,
superscriptaddress,
 amsmath,amssymb,
aps,
pra,
showpacs
]{revtex4-2}

\usepackage{bbm}
\usepackage{amsthm,amsmath,amssymb}
\usepackage{graphicx,comment}
\usepackage{braket}
\usepackage{hyperref}
\hypersetup{hidelinks}
\allowdisplaybreaks[2]

\theoremstyle{plain}
\newtheorem{theorem}{Theorem}[section]
\newtheorem{proposition}[theorem]{Proposition}
\newtheorem{lemma}[theorem]{Lemma}
\newtheorem{corollary}[theorem]{Corollary}
\theoremstyle{definition}
\newtheorem{definition}{Definition}[section]

\DeclareMathOperator{\conv}{conv}
\DeclareMathOperator{\rank}{rank}

\newcommand{\cP}{\mathcal{P}}

\newcommand{\C}{\mathbb{C}}
\newcommand{\R}{\mathbb{R}}

\newcommand{\Cc}{\mathcal{C}}
\newcommand{\Gg}{\mathcal{G}}
\newcommand{\Ss}{\mathcal{S}}
\newcommand{\Pn}{\widetilde{\mathcal{P}}_n}
\newcommand{\Pnb}{\overline{\mathcal{P}}_n}
\newcommand{\Phin}{\Phi_n}
\newcommand{\Gb}{\overline{\mathcal{G}}}
\newcommand{\Sb}{\overline{\mathcal{S}}}
\newcommand{\id}{\mathbbm{1}}

\newcommand{\Tr}{\operatorname{Tr}}

\renewcommand{\braket}[2]{\langle #1 \mid #2 \rangle}

\begin{document}

\title{Geometry of Knill--Laflamme Coefficients for Pauli Error Detection}

\author{Baisong Sun}
\affiliation{Department of Physics, The University of Texas at Dallas, Richardson, Texas 75080, USA}
\author{Ningping Cao}
\affiliation{Digital Technologies, National Research Council Canada, Canada}
\affiliation{Institute for Quantum Computing, University of Waterloo, ON N2L 3G1, Canada}
\author{Yiu Tung Poon}
\affiliation{Department of Mathematics, Iowa State University, Ames, Iowa 50011, USA}
\author{Bei Zeng}
\affiliation{Department of Physics, The University of Texas at Dallas, Richardson, Texas 75080, USA}

\begin{abstract}
Knill--Laflamme coefficients characterize exact quantum error detection through scalar compressions of error operators to the code space. For prescribed Pauli observables and a fixed code dimension, attainable coefficient vectors form a joint higher-rank numerical range, whose Euclidean norm image is the signature spectrum. We study how operator representation, code dimension, and detection constraints govern this geometry. Our main theorem identifies representation multiplicity as sufficient to lift active expectation data to exact scalar compressions: a lift exists whenever each occupied block has multiplicity at least the code dimension times the rank of the corresponding active state. When multiplicity suffices to realize all active states, the coefficient range is the full active expectation body and is therefore convex and connected, with a closed interval as its signature spectrum. This framework unifies commuting Pauli families, where sufficient common-eigenspace degeneracy yields polytopes, and subsystem stabilizer codes, where protected logical subsystems supply multiplicity spaces for noncommuting gauge observables. Explicit examples show how increasing the code dimension can make a range shrink, collapse, or become empty, while increasing multiplicity can turn a Bloch sphere into a filled ball. In the subsystem setting, parent-Hamiltonian constructions realize exact detecting codes as degenerate ground spaces, selecting distinguished vectors and continuous paths within these ranges. A final three-qubit example shows additional noncommuting detection constraints reduce a tetrahedral range to its barycenter and four vertices, producing a disconnected signature spectrum. The results provide a structural framework for Knill--Laflamme coefficient geometry and motivate sharper criteria for connected coefficient ranges and interval signature spectra under structured error models.
\end{abstract}

\maketitle


\section{Introduction}
\label{sec:intro}
\begin{figure*}[t!]
  \centering
  \includegraphics[width=\textwidth]{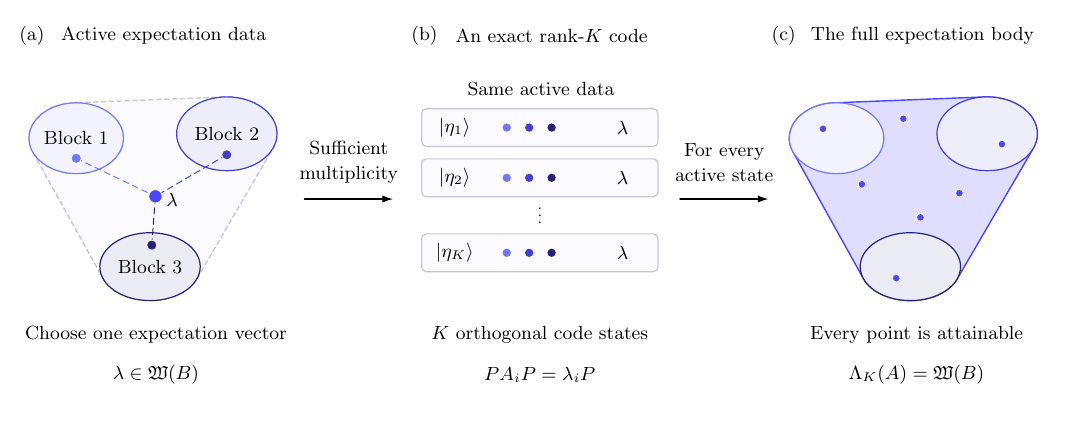}
  \caption{Representation multiplicity as a mechanism for exact scalar
  compression. (a) Expectations from the active blocks combine to give a
  selected vector $\lambda$ in the active state-space expectation body
  $\mathfrak W(B)$. (b) Under the state-dependent sufficient condition of
  Theorem~\ref{thm:block-multiplicity}, multiplicity permits a realization
  by $K$ orthogonal code states with the same active data; every observable
  compresses to a scalar on their span. (c) Under the theorem's uniform sufficient
  condition, this construction realizes every point of $\mathfrak W(B)$.
  The compression range is then convex and connected, and its signature
  spectrum is a closed interval. The shapes are schematic; the colored
  marks in (b) represent the same active block data in each code state.}
  \label{fig:intro-block-multiplicity}
\end{figure*}

Quantum error correction protects information against physical noise by encoding logical states into a larger Hilbert space.  In the exact setting, a code must perfectly detect a prescribed error family.  For Pauli models, this may be the full set of errors up to a chosen weight or a restricted family adapted to biased or structured noise~\cite{IoffeMezard2007}.  Fixing the detectable family and the code dimension therefore defines a geometric feasibility problem for rank-$K$ code subspaces.

Important solutions arise from stabilizer codes~\cite{LaflammeMiquelPazZurek1996,Gottesman1997,CalderbankRainsShorSloane1997}, nonadditive codes~\cite{RainsHardinShorSloane1997}, and codeword-stabilized codes~\cite{CrossSmithSmolinZeng2009,ChuangCrossSmithSmolinZeng2009,BeigiChenGrasslJiWangZeng2013}; permutation-invariant codes~\cite{Ouyang2014,OuyangFitzsimons2016}, topological codes~\cite{DennisKitaevLandahlPreskill2002}, and subsystem codes~\cite{KribsLaflammePoulin2005,Poulin2005} further broaden this landscape.  These constructions answer important existence questions and provide distinguished families, but they do not by themselves describe the geometry of all coefficient data attainable in a fixed detection problem.  The feasible projectors lie in a nonconvex manifold, and their compression profiles may form polyhedral, curved, lower-dimensional, finite, or disconnected sets.  Our aim is to identify structural mechanisms that organize this geometry.

The Knill--Laflamme conditions turn the problem into one of simultaneous matrix compression.  For a prescribed family \(\mathcal E\) of Hermitian Pauli operators, exact detection by a rank-\(K\) code projector \(P\) requires
\begin{equation}
P E P=\alpha_E P,
\qquad E\in\mathcal E.
\end{equation}
Thus no operator in the detectable family carries logical-state-dependent information inside the code.  For a Pauli noise model, the general correction condition is recovered by imposing the same detection equation on Hermitian representatives of the products \(E_a^\dagger E_b\)~\cite{KnillLaflamme1997}.  Exact Pauli detection therefore asks for a common rank-\(K\) subspace on which an entire Hermitian Pauli family compresses to scalars.  This formulation applies to ordinary subspace codes, including those obtained from subsystem codes by fixing the gauge subsystem to a pure state while retaining the full logical subsystem.

For an ordered Hermitian tuple \(A=(A_1,\ldots,A_m)\), the set of all simultaneously achievable compression coefficients is the \emph{joint higher-rank numerical range}
\begin{equation}
\Lambda_K(A)
=
\left\{
\lambda\in\R^m:
\begin{array}{l}
\exists\, P\in\cP_K \text{ such that}\\
P A_j P=\lambda_j\,P,\quad j=1,\dots,m
\end{array}
\right\}.
\label{eq:jnr-def}
\end{equation}
Its nonemptiness decides whether a \(K\)-dimensional code detecting the prescribed operators exists.  Its geometry describes which Knill--Laflamme coefficient profiles can occur, whether the attainable profiles form a connected continuum, and how their shape changes with the error family, compression rank, and available degeneracy.  This profile geometry must be distinguished from the topology of the realizing projector set: connectedness of \(\Lambda_K(A)\) does not by itself imply that all realizing codes lie in one connected family.  Joint higher-rank numerical ranges nevertheless provide a natural geometric language for studying the landscape of exact code data rather than only existence.

This connection has motivated a substantial operator-theoretic literature.  Beginning with the compression formulation of quantum error correction, this line of work developed constructions of correctable codes and established convexity and spectral descriptions for higher-rank numerical ranges, together with extensions to joint and generalized ranges and to perturbations of quantum channels~\cite{ChoiKribsZyczkowski2006,ChoiKribsZyczkowskiCompression2006,Woerdeman2008,LiSze2008,GauLiPoonSze2011,LiPoon2011,LiPoonSze2008}.  For Pauli observables, related work has studied joint higher-rank ranges and, at rank one, expectation and variance bodies governed by graph parameters~\cite{AfshinMehrjoofardZare2018,XuSchwonnekWinter2024,XuWangYeKossmannSchwonnekWinter2025}.  The single-operator theory enjoys strong convexity and spectral structure, but for tuples with more than two coordinates the geometry can be substantially more complicated.  Even for small Pauli tuples, joint rank-$K$ ranges can fail to be convex and can be disconnected.

The full compression vector is not the only useful global summary of this landscape.  Alongside the vector range, we consider its radial image
\begin{equation}
\Sigma_K(A):=\{\,\|\lambda\|_2:\lambda\in\Lambda_K(A)\,\},
\end{equation}
which we call the \emph{signature spectrum}.  For an ordered Pauli tuple \(E=(E_1,\ldots,E_m)\) and a detecting projector \(P\), the corresponding scalar is
\begin{equation}
\lambda^{*}(P)
=
\left(
\sum_{\alpha=1}^{m}
\langle E_\alpha\rangle_{\rho_P}^{\,2}
\right)^{1/2},
\label{eq:lambda-star}
\end{equation}
where $\rho_P :=P/K$ and $\langle E_\alpha\rangle_{\rho_P} := \Tr(E_\alpha \rho_P)$. 
Because a Hermitian Pauli operator satisfies \(E_\alpha^2=\id\), the same squared expectation values determine its variance on the maximally mixed code state.  The connectedness of \(\Lambda_K(A)\) and that of \(\Sigma_K(A)\) are distinct questions: a disconnected vector range may still have a connected norm image if its components occur at the same or overlapping radii.

Recently, the signature vector and its norm were introduced as invariants of Knill--Laflamme coefficient data in Ref.~\cite{DuZhangPoonZeng2026}.  Continuous families of \(((6,2,3))\) and \(((7,2,3))\) codes were shown to realize closed intervals of signature values.  The same work also gave a two-qubit Pauli tuple whose joint rank-two numerical range is disconnected while its signature spectrum is the singleton \(\{1\}\), demonstrating that vector connectedness and radial connectedness are distinct geometric properties.  The radial geometry of Pauli-detecting problems was studied further in Ref.~\cite{GuptaSunHeZeng2026}, where interval signature spectra were observed in the unrestricted Pauli-detection problems studied there.  Since the signature spectrum is the Euclidean norm image of the full compression range, this interval behavior is naturally part of the broader geometry of Pauli compression vectors.

This motivates the central questions of the present work: \emph{which representation-theoretic and rank-theoretic structures govern the geometry of Pauli compression vectors?  When can active expectation data be lifted to exact rank-$K$ scalar compressions, under what conditions is the resulting range convex or connected, and what geometric obstructions arise from insufficient degeneracy or additional detection constraints?}

Our main structural result gives a sufficient condition for lifting active expectation data to exact scalar compressions. In a suitable basis, the observables share a block structure in which the active operators, collectively denoted by $B$, are tensored with identities on multiplicity spaces. We write $\mathfrak W(B)$ for the convex hull of the expectation vectors obtained from all states on the individual active factors. Figure~\ref{fig:intro-block-multiplicity}(a) shows how expectations from the active blocks combine to determine a target vector in $\mathfrak W(B)$. When the multiplicity spaces are sufficiently large, the same active data can be realized by $K$ orthogonal code states, on whose span every prescribed observable compresses to a scalar, as shown in Fig.~\ref{fig:intro-block-multiplicity}(b). If sufficient multiplicity is available for every active state, the construction realizes the entire expectation body $\mathfrak W(B)$, as shown in Fig.~\ref{fig:intro-block-multiplicity}(c). The compression range is then convex and connected, and its signature spectrum is a closed interval.

For commuting Pauli families, this principle yields polytopes determined by the joint eigenvalues, with explicit projectors realizing every point whenever the common eigenspaces are sufficiently degenerate. Beyond this regime, increasing the compression rank can change the geometry sharply. For the three-qubit tuple $(Z_1,Z_2,Z_3)$, the range passes from a cube to an octahedron, then to the origin, and finally to the empty set. Thus the available degeneracy determines when the full polytope can be attained, while larger code dimensions can impose further geometric restrictions.

For noncommuting gauge observables, protected logical subsystems supply natural multiplicity spaces. Our subsystem examples illustrate how additional multiplicity enlarges the attainable range: a Bloch sphere becomes a filled Bloch ball as mixtures of pure-state expectation data become realizable by exact codes. Parent Hamiltonians give this structure a complementary interpretation through ground spaces whose compression vectors select distinguished points or trace continuous paths within these ranges.

Finally, additional detection requirements can obstruct these realization mechanisms. We construct a three-qubit example in which noncommuting Pauli constraints reduce a tetrahedral compression range to its barycenter and four vertices. The resulting signature spectrum is $\{0,\sqrt{3}\}$, showing that the interval behavior observed in earlier work~\cite{GuptaSunHeZeng2026} is not universal for prescribed finite Pauli families. Whether it persists for more structured error families, such as those defined by a weight cutoff, remains open.

The remainder of the paper follows this structural organization.  Section~\ref{sec:prelim} fixes the compression framework and the notation for Pauli detection, signature spectra, and subsystem codes.  Section~\ref{sec:multiplicity} establishes the block-multiplicity theorem and its active/spectator consequences.  Section~\ref{sec:commuting} develops the commuting specialization, deriving common-sector polytopes, explicit realizing projectors, and rank-sensitive examples.  Section~\ref{sec:subsystem} identifies protected logical subsystems as multiplicity spaces for noncommuting gauge algebras, and Sec.~\ref{sec:hamiltonians} gives this branch a complementary parent-Hamiltonian interpretation.  Section~\ref{sec:collapse} presents a three-qubit example whose signature spectrum is disconnected, and discusses what it implies for the interval observation.  Section~\ref{sec:discussion} synthesizes the hierarchy of geometric mechanisms and discusses the scope of the remaining open questions, while the appendices collect supporting computations and proofs.


\section{Preliminaries}
\label{sec:prelim}

This section fixes the notation used throughout the paper.  We first specify the ambient $n$-qubit operator setting, Hermitian-tuple conventions, and code-state notation.  We then recall ordinary and background-constrained joint higher-rank numerical ranges, together with the normalized-trace and semialgebraic facts used later.  Finally, we introduce Pauli notation, relate scalar compression to the Knill--Laflamme detection condition, and fix our conventions for subsystem stabilizer codes.

\subsection{Operator and projector notation}
\label{subsec:operator-projector-notation}

We work on the $n$-qubit Hilbert space
\begin{equation}
\mathcal H_n \cong (\C^2)^{\otimes n},\qquad \dim\mathcal H_n=2^n,
\end{equation}
and write $\mathcal B(\mathcal H_n)$ for the algebra of linear operators on $\mathcal H_n$.  We write $\id_N$ for the identity on an $N$-dimensional Hilbert space and write simply $\id$ when the dimension is clear.
All numerical ranges below are formed from Hermitian operator tuples $A=(A_1,\ldots,A_m)$; reordering the tuple merely permutes the coordinates of the corresponding subset of $\R^m$.

A quantum code is a $K$-dimensional subspace $\Cc\subset\mathcal H_n$, and $P$ denotes the orthogonal projector onto $\Cc$.  For any finite-dimensional Hilbert space $\mathcal H$ and $1\le K\le\dim\mathcal H$, define
\begin{equation}
\cP_K(\mathcal H):=\{P\in\mathcal B(\mathcal H):\ P=P^\dagger=P^2,\ \rank P=K\}.
\end{equation}
Thus $P\in\cP_K(\mathcal H_n)$ and $\Cc=\operatorname{Ran}(P)$.  When the ambient Hilbert space is clear from context, we abbreviate $\cP_K(\mathcal H)$ to $\cP_K$; in particular, in the $n$-qubit setting $\cP_K$ means $\cP_K(\mathcal H_n)$.  The maximally mixed state supported on the code space is
\begin{equation}
\rho_P:=\frac{P}{K}
\end{equation}
and serves as a basis-independent representative of the code.  For any state $\rho$ and operator $A$ we use the expectation notation
\begin{equation}
\langle A\rangle_\rho:=\Tr(\rho A).
\end{equation}
For a unit vector $\ket{\psi}$, we likewise write
\begin{equation}
\langle A\rangle_\psi:=\langle\psi|A|\psi\rangle.
\end{equation}
If $\Pi$ is another orthogonal projector, the notation $P\le \Pi$ means $\operatorname{Ran}(P)\subseteq \operatorname{Ran}(\Pi)$, equivalently $P=\Pi P=P\Pi$.

\subsection{Joint higher-rank numerical ranges and scalar compressions}
\label{subsec:jhrnr}

\begin{definition}[Numerical range]
For a Hermitian operator $A\in\mathcal B(\mathcal H_n)$, its numerical range is
\begin{equation}
W(A)
:=
\left\{
\langle A\rangle_\psi:
\begin{array}{l}
\psi\in\mathcal H_n,\\
\|\psi\|=1
\end{array}
\right\}
\subseteq\R.
\end{equation}
Thus $W(A)$ is the set of expectation values attainable on pure states.
\end{definition}

\begin{definition}[Joint numerical range]
For an $m$-tuple of Hermitian operators $A=(A_1,\dots,A_m)$ on $\mathcal H_n$, the joint numerical range is
\begin{equation}
\Lambda_1(A)
:=
\left\{
\bigl(\langle A_1\rangle_\psi,\dots,\langle A_m\rangle_\psi\bigr):
\begin{array}{l}
\psi\in\mathcal H_n,\\
\|\psi\|=1
\end{array}
\right\}
\subseteq\R^m.
\end{equation}
For $m=1$, this agrees with the numerical range $W(A_1)$.
\end{definition}

\begin{definition}[Joint higher-rank numerical range]
Let $A=(A_1,\dots,A_m)$ be an $m$-tuple of Hermitian operators on $\mathcal H_n$.  For $1\le K\le 2^n$, the joint rank-$K$ numerical range of $A$ is
\begin{equation}
\Lambda_K(A)
:=
\left\{
\lambda\in\R^m:
\begin{array}{l}
\exists\, P\in\cP_K \text{ such that}\\
P A_j P=\lambda_j\,P,\quad j=1,\dots,m
\end{array}
\right\}.
\label{eq:LambdaK-prelim}
\end{equation}
\end{definition}

The condition $P A_j P=\lambda_j P$ says that the compression of $A_j$ to the subspace $\operatorname{Ran}(P)$ is the scalar operator $\lambda_j\id_{\operatorname{Ran}(P)}$.  Thus a point of $\Lambda_K(A)$ records a $K$-dimensional subspace on which all operators in the tuple have simultaneous scalar compressions.  Taking normalized traces gives
\begin{equation}
\lambda_i = \Tr(\rho_P A_i),
\label{eq:lambda-expectation}
\end{equation}
so every scalar-compression vector is an expectation vector of the maximally mixed code state.  When $K=1$, writing $P=\ket{\psi}\!\bra{\psi}$ identifies $\Lambda_K(A)$ with the joint numerical range $\Lambda_1(A)$.  This compression viewpoint underlies the convexity and spectral theory of single-operator higher-rank numerical ranges, as well as structural studies of their joint counterparts~\cite{ChoiKribsZyczkowski2006,LiPoonSze2008,GauLiPoonSze2011,Woerdeman2008,AfshinMehrjoofardZare2018,LiSze2008}.

The unconstrained construction is unitarily invariant.  For any unitary $V$ on $\mathcal H_n$,
\begin{equation}
\Lambda_K(A_1,\ldots,A_m)
=
\Lambda_K(VA_1V^\dagger,\ldots,VA_mV^\dagger),
\end{equation}
with compressing projectors related by $P\mapsto VPV^\dagger$.  This elementary covariance is used repeatedly when passing to Clifford normal forms.

\subsection{Background-constrained compression ranges}
\label{subsec:constrained-ranges}

The ordinary joint higher-rank numerical range $\Lambda_K(A)$ in Eq.~\eqref{eq:LambdaK-prelim} uses the same tuple both to impose the scalar-compression equations and to record their scalar coordinates.  In an error-detection problem it is often useful to separate the tuple whose coordinates are monitored from a larger family that the code is required to detect.  The scalar-signature formulation treats these two roles together~\cite{GuptaSunHeZeng2026}; separating them explicitly will also let us describe coordinate projections of a full joint range without discarding the background constraints, namely, the scalar-compression conditions imposed by that family.

For a finite Hermitian family $\mathcal F$ on $\mathcal H_n$, define its rank-$K$ detecting-projector set by
\begin{equation}
\operatorname{Det}_K(\mathcal F)
:=
\left\{
\begin{array}{l}
P\in\cP_K(\mathcal H_n):\\
PFP=K^{-1}\Tr(PF)P\quad\forall F\in\mathcal F
\end{array}
\right\}.
\label{eq:DetK-background}
\end{equation}
Let $A=(A_1,\ldots,A_m)$ be an ordered monitored tuple, not necessarily contained in $\mathcal F$.

\begin{definition}[Background-constrained joint range]
\label{def:background-range}
The rank-$K$ joint range of $A$ under the background detection constraints $\mathcal F$ is
\begin{equation}
\begin{aligned}
&\Lambda_K(A\mid\mathcal F)\\
&\quad:=
\left\{
\lambda\in\R^m:
\begin{array}{l}
\exists\,P\in\operatorname{Det}_K(\mathcal F),\\
PA_iP=\lambda_iP,\ i=1,\ldots,m
\end{array}
\right\}.
\end{aligned}
\label{eq:background-range}
\end{equation}
Its \emph{signature spectrum}, namely its Euclidean norm image, is
\begin{equation}
\Sigma_K(A\mid\mathcal F)
:=
\{\,\|\lambda\|_2:\lambda\in\Lambda_K(A\mid\mathcal F)\,\}.
\label{eq:background-spectrum}
\end{equation}
\end{definition}

For an unconstrained ordered tuple $A$, we use the analogous abbreviation
\begin{equation}
\Sigma_K(A)
:=
\{\,\|\lambda\|_2:\lambda\in\Lambda_K(A)\,\}.
\label{eq:ordinary-spectrum}
\end{equation}

If $\mathcal F=\{A_1,\ldots,A_m\}$, then
\begin{equation}
\Lambda_K(A\mid\mathcal F)=\Lambda_K(A).
\end{equation}
More generally, if an ordered full tuple $A^{\mathrm{full}}$ lists every member of $\mathcal F$ and contains $A$ as a subtuple, then
\begin{equation}
\Lambda_K(A\mid\mathcal F)
=
\pi_A\bigl(\Lambda_K(A^{\mathrm{full}})\bigr),
\label{eq:background-projection}
\end{equation}
where $\pi_A$ is the coordinate projection retaining the monitored entries.  Thus the constrained range is the visible coordinate projection of the full range after all detection equations have been imposed.  The vector range is the finer object used to analyze the signature spectrum in Eq.~\eqref{eq:background-spectrum}; the two must be distinguished because they can have different connectedness behavior.

\begin{proposition}[Compact semialgebraic structure]
\label{prop:semialgebraic-components}
For every finite Hermitian family $\mathcal F$ and monitored Hermitian tuple $A$, the set of rank-$K$ projectors that detect $\mathcal F$ and compress $A$ scalarly is compact and semialgebraic.  Consequently, $\Lambda_K(A\mid\mathcal F)$ is a compact semialgebraic subset of $\R^m$ and has finitely many connected components.  If the feasible projector set has connected components $\mathcal D_1,\ldots,\mathcal D_N$, then for each $a$ there are $0\le r_a^-\le r_a^+$ such that
\begin{equation}
\left\{
\left\|
\left(\frac1K\Tr(PA_i)\right)_{i=1}^m
\right\|_2:
P\in\mathcal D_a
\right\}
=
[r_a^-,r_a^+].
\label{eq:component-radial-interval}
\end{equation}
In particular,
\begin{equation}
\Sigma_K(A\mid\mathcal F)
=
\bigcup_{a=1}^N[r_a^-,r_a^+].
\label{eq:finite-union-radial-intervals}
\end{equation}
\end{proposition}

\begin{proof}
The rank-$K$ projector manifold is compact and real algebraic.  For each $F\in\mathcal F$, scalar compression can be written without auxiliary scalar variables as
\begin{equation}
PFP-\frac1K\Tr(PF)P=0,
\label{eq:polynomial-compression}
\end{equation}
and similarly for the monitored operators.  These are polynomial equations in the real and imaginary parts of the entries of $P$.  The feasible set is therefore closed and semialgebraic, hence compact and with finitely many connected components.  The coordinate map
\begin{equation}
P\longmapsto
\left(\frac1K\Tr(PA_i)\right)_{i=1}^m
\end{equation}
is polynomial, so its image is semialgebraic by the Tarski--Seidenberg theorem and compact by continuity.  On each connected component $\mathcal D_a$, composition with the Euclidean norm is continuous; its image is therefore a compact connected subset of $\R$, hence the interval in Eq.~\eqref{eq:component-radial-interval}.  Taking the union over these finitely many connected components yields Eq.~\eqref{eq:finite-union-radial-intervals}.
\end{proof}

Thus disconnectedness of the vector range does not by itself imply a disconnected signature spectrum; a counterexample to interval behavior requires a gap between the component intervals in Eq.~\eqref{eq:finite-union-radial-intervals}.

\subsection{Pauli notation and code detection}
\label{subsec:pauli-code-notation}

We first relate the scalar-compression language above to the usual Knill--Laflamme condition.  For a finite error model \(\mathcal N=\{N_\alpha\}\subset\mathcal B(\mathcal H_n)\), a rank-\(K\) projector \(P\) corrects \(\mathcal N\) precisely when there is a Hermitian matrix \(c=(c_{\alpha\beta})\) such that
\begin{equation}
P N_\alpha^\dagger N_\beta P=c_{\alpha\beta}P,
\qquad \forall\,\alpha,\beta .
\label{eq:KL-general-prelim}
\end{equation}
This is the Knill--Laflamme condition~\cite{KnillLaflamme1997}.
To express this condition in numerical-range language, choose a Hermitian operator tuple \(A=(A_1,\ldots,A_m)\) spanning the real linear space generated by the operators
\begin{equation}
\frac{N_\alpha^\dagger N_\beta+N_\beta^\dagger N_\alpha}{2},
\qquad
\frac{N_\alpha^\dagger N_\beta-N_\beta^\dagger N_\alpha}{2i}.
\end{equation}
Then Eq.~\eqref{eq:KL-general-prelim} holds if and only if
\begin{equation}
P A_j P=\lambda_jP,\qquad j=1,\ldots,m,
\label{eq:KL-Hermitian-prelim}
\end{equation}
for some real scalars \(\lambda_1,\ldots,\lambda_m\), equivalently if and only if \((\lambda_1,\ldots,\lambda_m)\in\Lambda_K(A)\).  Thus the nonemptiness of an appropriate joint higher-rank numerical range is the Hermitian-operator form of the Knill--Laflamme correction condition~\cite{ChoiKribsZyczkowskiCompression2006,LiPoon2011}.

We now specialize this bridge to Pauli error models.  To keep operator phases separate from binary Pauli labels, let
\begin{equation}
\Pn
:=
\bigl\{i^a P_1\otimes\cdots\otimes P_n:
a\in\mathbb Z_4,\ P_j\in\{I,X,Y,Z\}\bigr\}
\end{equation}
be the full \(n\)-qubit Pauli group, including scalar phases, and let
\begin{equation}
\Phin:=\{i^a\id_{2^n}:a\in\mathbb Z_4\},
\qquad
\Pnb:=\Pn/\Phin\cong\mathbb F_2^{2n}
\label{eq:pauli-phase-quotient}
\end{equation}
be its phase-free label space.  We write \(\pi:\Pn\to\Pnb\) for the quotient map.  Independence and binary ranks always refer to labels in \(\Pnb\), and commutation is recorded by its standard symplectic form \(\omega\): Pauli operators \(P,Q\) commute exactly when \(\omega(\pi(P),\pi(Q))=0\)~\cite{CalderbankRainsShorSloane1997,Gottesman1997}.  Write
\begin{equation}
\Pn^{\mathrm H}:=\{P\in\Pn:P=P^\dagger\}
\end{equation}
for the Hermitian Pauli operators.  A \emph{Pauli observable} means an element of \(\Pn^{\mathrm H}\), equivalently a tensor product of \(I,X,Y,Z\) up to an overall sign. When the register size is fixed, we use indexed embedding notation in structural statements and Pauli-word notation in explicit operator lists.  When the entries of a Hermitian tuple are specific Pauli observables, we often write \(E=(E_1,\ldots,E_m)\) to distinguish this Pauli tuple from a general Hermitian tuple \(A\).

For a Hermitian Pauli error family \(\mathcal E\subset\Pn^{\mathrm H}\), the detection condition takes the simpler scalar-compression form
\begin{equation}
P E P=\alpha_E P,\qquad \forall\,E\in\mathcal E .
\label{eq:pauli-detection-prelim}
\end{equation}
Whenever this holds, \(P\) is an \(\mathcal E\)-detecting projector and \(\operatorname{Ran}(P)\) is an \(\mathcal E\)-detecting code.  If an ordered Pauli tuple \(E=(E_1,\ldots,E_m)\) is monitored, the feasible scalar tuples are precisely the points of \(\Lambda_K(E)\).  In the standard distance-based setting, one takes
\begin{equation}
\mathcal E_d:=\{E\in\Pn^{\mathrm H}:\ \operatorname{wt}(E)<d\},
\end{equation}
where the Pauli weight $\operatorname{wt}(E)$ is the number of tensor factors of $E$ that are different from the single-qubit identity $I$. An \(((n,K,d))\) quantum code on \(n\) physical qubits is a \(K\)-dimensional subspace \(\Cc\subset\mathcal H_n\) whose projector \(P\) detects \(\mathcal E_d\). More generally, the detectable Pauli family need not arise from a single uniform weight cutoff; it may instead be prescribed directly. Biased or asymmetric noise provides one important example~\cite{IoffeMezard2007,GuptaSunHeZeng2026}.

\subsection{Subsystem stabilizer codes}
\label{subsec:subsystem-notation}

A Pauli subsystem stabilizer code introduces an additional structure by splitting the code subspace into a protected logical factor and an unprotected gauge factor~\cite{KribsLaflammePoulin2005,Poulin2005,Bacon2006}.  We write its parameters as \([[n,k;t,s]]\): \(n\) physical qubits, \(k\) logical qubits, \(t\) gauge qubits, and \(s\) independent phase-free stabilizer labels. Thus
\begin{equation}
s=n-k-t,
\qquad
\dim\Cc=2^{k+t}.
\end{equation}
Such a code is specified by an operator gauge group \(\Gg\le\Pn\) containing the scalar phase subgroup \(\Phin\).  Its phase-free gauge-label space and symplectic radical are
\begin{equation}
\Gb:=\pi(\Gg)\subseteq\Pnb,
\qquad
\Sb:=\Gb\cap\Gb^{\perp_\omega}.
\label{eq:gauge-label-radical}
\end{equation}
Choose a basis of \(\Sb\), choose mutually commuting Hermitian lifts \(S_1,\ldots,S_s\in\Gg\), and fix their signs.  The resulting phase-fixed stabilizer subgroup is
\begin{equation}
\Ss:=\langle S_1,\ldots,S_s\rangle,
\qquad -\id\notin\Ss .
\label{eq:phase-fixed-stabilizer}
\end{equation}
The inclusion \(\Phin\subseteq\Gg\) ensures that the operator gauge group contains \(-\id\), independently of the number \(t\) of gauge qubits. When at least one gauge qubit is present, \(-\id\) also arises as the group commutator of an anticommuting gauge pair. By contrast, \(-\id\) is excluded from the phase-fixed stabilizer subgroup \(\Ss\), whose elements impose the code-space eigenvalue constraints. Thus the scalar phases remain in the operator center of the gauge group but not in the stabilizer constraints; more precisely,
\begin{equation}
Z(\Gg)=\Phin\!\cdot\!\Ss,
\qquad
\pi\bigl(Z(\Gg)\bigr)=\Sb.
\label{eq:gauge-center-phases}
\end{equation}
The code subspace is the common \(+1\) eigenspace of the chosen phase-fixed stabilizers,
\begin{equation}
\Cc
=
\{\,\ket{\psi}\in\mathcal H_n:\ S\ket{\psi}=\ket{\psi}\ \forall S\in\Ss\,\}.
\end{equation}
Here \(s=\dim_{\mathbb F_2}\Sb\). The subsystem structure gives a noncanonical factorization
\begin{equation}
\Cc\cong\mathcal H_G\otimes\mathcal H_L,\qquad
\dim\mathcal H_L=2^k,\qquad \dim\mathcal H_G=2^t,
\end{equation}
where $\mathcal H_L$ is the protected logical factor and $\mathcal H_G$ is the gauge factor.

After choosing independent labels and Hermitian lifts, the gauge group has phase-free dimension \(s+2t\) and can be written in the standard symplectic paired form
\begin{equation}
\label{eq:gauge-generators-abstract}
\Gg
=
\big\langle
i\id,\ S_1,\dots,S_s,\ X_{g_1},Z_{g_1},\ \dots,\ X_{g_t},Z_{g_t}
\big\rangle,
\end{equation}
where the stabilizer generators $S_i$ commute with all generators, each pair $X_{g_j},Z_{g_j}$ anticommutes, and all other listed pairs commute.  Thus
\begin{equation}
\Ss=\langle S_1,\dots,S_s\rangle,\qquad
\dim_{\mathbb F_2}\Gb=s+2t.
\end{equation}
The canonical-form lemma in Sec.~\ref{subsec:canonical} makes this decomposition concrete: after a Clifford change of basis, every gauge operator has the form \(O_g\otimes\id_{2^k}\), with the logical factor supplying a \(2^k\)-fold multiplicity. This explains the natural compression rank \(2^k\) in the subsystem setting. More generally, the trace identity in Eq.~\eqref{eq:lambda-expectation} is necessary for scalar compression, but it does not by itself imply the operator identities \(PA_iP=\lambda_iP\). The next section shows how block multiplicity provides a sufficient mechanism for lifting state-space expectation vectors to exact rank-\(K\) scalar compressions.


\section{Main structural theorem: representation multiplicity}
\label{sec:multiplicity}

The principal positive realization results in this paper are governed by one representation-theoretic mechanism.  A rank-$K$ scalar compression requires $K$ orthonormal vectors on which every operator in the tuple has identical diagonal matrix elements and vanishing off-diagonal matrix elements.  When the tuple acts through active matrix blocks tensored with multiplicity spaces, the active factors determine the expectation vector, while the multiplicity factors provide orthogonal labels for these $K$ copies.  The compression problem therefore becomes a rank-budget problem on the multiplicity spaces.

The main theorem first places every scalar-compression vector inside the state-space expectation body of the active block algebra.  It then gives a sufficient and explicit multiplicity condition for lifting an active block state to an exact rank-$K$ scalar compression.  Two immediate specializations prepare the rest of the paper: one-dimensional active blocks give the convex hull of common spectral data, while a single active block gives the active/spectator sandwich and the two thresholds for pure-state realization and full convexification.

\subsection{Active blocks and multiplicity spaces}
\label{subsec:block-multiplicity}

Let
\begin{equation}
\mathcal H
=
\bigoplus_{\alpha=1}^{\nu}
\left(\mathcal H_\alpha\otimes\C^{m_\alpha}\right),
\label{eq:block-multiplicity-H}
\end{equation}
and suppose that the Hermitian tuple $A=(A_1,\ldots,A_m)$ has the block form
\begin{equation}
A_i
=
\bigoplus_{\alpha=1}^{\nu}
\left(B_{i,\alpha}\otimes\id_{m_\alpha}\right),
\qquad i=1,\ldots,m.
\label{eq:block-multiplicity-A}
\end{equation}
Write
\begin{equation}
B:=\bigl(B_{i,\alpha}\bigr)_{\substack{1\le i\le m\\1\le\alpha\le \nu}}
\end{equation}
for this block family.  Let $\mathcal D(\mathcal K)$ denote the density matrices on a finite-dimensional Hilbert space $\mathcal K$.  To parameterize the weights assigned to the active blocks, for every nonnegative integer $q$ let
\begin{equation}
\Delta_q
:=
\left\{
p\in\R^{q+1}:
\begin{array}{l}
p_j\ge 0,\quad j=1,\ldots,q+1,\\
\displaystyle\sum_{j=1}^{q+1}p_j=1
\end{array}
\right\}
\label{eq:probability-simplex}
\end{equation}
denote the standard $q$-dimensional probability simplex.  We then define the state-space expectation body
\begin{equation}
\mathfrak W(B)
:=
\left\{
\left(
\sum_{\alpha=1}^{\nu}
p_\alpha\Tr(\rho_\alpha B_{i,\alpha})
\right)_{i=1}^{m}:
\begin{array}{l}
p\in\Delta_{\nu-1},\\
\rho_\alpha\in\mathcal D(\mathcal H_\alpha)
\end{array}
\right\}.
\label{eq:block-state-space-body}
\end{equation}
This is the joint expectation body of the finite-dimensional algebra $\bigoplus_\alpha\mathcal B(\mathcal H_\alpha)$.  The next theorem is the main structural result of the paper.

\begin{theorem}[Block multiplicity and exact scalar compression]
\label{thm:block-multiplicity}
For the tuple in Eq.~\eqref{eq:block-multiplicity-A}, one always has
\begin{equation}
\Lambda_K(A)\subseteq\mathfrak W(B).
\label{eq:block-multiplicity-upper}
\end{equation}
Conversely, fix $p\in\Delta_{\nu-1}$ and density matrices $\rho_\alpha$ on $\mathcal H_\alpha$.  If
\begin{equation}
m_\alpha\ge K\,\rank(\rho_\alpha)
\label{eq:block-rank-condition}
\end{equation}
for every $\alpha$ with $p_\alpha>0$, then the corresponding point of $\mathfrak W(B)$ belongs to $\Lambda_K(A)$. In particular, if
\begin{equation}
m_\alpha\ge K\dim\mathcal H_\alpha
\label{eq:block-full-multiplicity-condition}
\end{equation}
for every $\alpha$, then
\begin{equation}
\Lambda_K(A)=\mathfrak W(B).
\label{eq:block-multiplicity-equality}
\end{equation}
The common set is compact, convex, and connected.
\end{theorem}

\begin{proof}
Suppose $P$ is a rank-$K$ scalar compression with $PA_iP=\lambda_iP$, and set $\rho=P/K$.  Let $\Pi_\alpha$ be the projector onto the $\alpha$th block in Eq.~\eqref{eq:block-multiplicity-H}.  Since each $A_i$ is block diagonal with respect to this decomposition, its expectation depends only on the diagonal blocks $\Pi_\alpha\rho\Pi_\alpha$ of $\rho$.  Define the positive subnormalized active operator
\begin{equation}
\tau_\alpha
:=
\Tr_{\C^{m_\alpha}}
\bigl(\Pi_\alpha\rho\Pi_\alpha\bigr),
\qquad
p_\alpha:=\Tr\tau_\alpha.
\end{equation}
Then $p\in\Delta_{\nu-1}$, and for $p_\alpha>0$ write $\tau_\alpha=p_\alpha\rho_\alpha$.  The defining property of the partial trace gives
\begin{equation}
\lambda_i
=\Tr(\rho A_i)
=\sum_{\alpha=1}^{\nu}
\Tr(\tau_\alpha B_{i,\alpha})
=\sum_{\alpha=1}^{\nu}
p_\alpha\Tr(\rho_\alpha B_{i,\alpha}),
\end{equation}
which proves Eq.~\eqref{eq:block-multiplicity-upper}.

For the converse, write
\begin{equation}
\rho_\alpha
=
\sum_{a=1}^{r_\alpha}
q_{\alpha a}\ket{\phi_{\alpha a}}\!\bra{\phi_{\alpha a}},
\qquad
r_\alpha=\rank(\rho_\alpha).
\label{eq:block-density-spectral}
\end{equation}
The condition in Eq.~\eqref{eq:block-rank-condition} permits, for every used block, a family of orthonormal spectator vectors
\begin{equation}
\left\{
\ket{e_{\alpha a j}}:
1\le a\le r_\alpha,\ 1\le j\le K
\right\}
\subset\C^{m_\alpha}.
\end{equation}
For $j=1,\ldots,K$, define
\begin{equation}
\ket{\eta_j}
:=
\bigoplus_{\alpha:p_\alpha>0}
\sqrt{p_\alpha}
\sum_{a=1}^{r_\alpha}
\sqrt{q_{\alpha a}}
\ket{\phi_{\alpha a}}\otimes\ket{e_{\alpha a j}}.
\label{eq:block-eta-construction}
\end{equation}
The normalization of the weights and orthogonality of the spectator vectors give
\begin{equation}
\langle\eta_j\mid\eta_\ell\rangle
=
\delta_{j\ell}
\sum_{\alpha:p_\alpha>0}p_\alpha
\sum_{a=1}^{r_\alpha}q_{\alpha a}
=\delta_{j\ell},
\label{eq:block-eta-orthonormal}
\end{equation}
so $\eta_1,\ldots,\eta_K$ are orthonormal.  The same spectator orthogonality also gives
\begin{equation}
\bra{\eta_j}A_i\ket{\eta_\ell}
=
\delta_{j\ell}
\sum_{\alpha=1}^{\nu}
p_\alpha\Tr(\rho_\alpha B_{i,\alpha}).
\label{eq:block-eta-compression}
\end{equation}
Set $P_\eta:=\sum_{j=1}^K\ket{\eta_j}\!\bra{\eta_j}$ and $\lambda_i:=\sum_\alpha p_\alpha\Tr(\rho_\alpha B_{i,\alpha})$.  Equation~\eqref{eq:block-eta-compression} then implies $P_\eta A_iP_\eta=\lambda_iP_\eta$ for every $i$, proving the converse assertion.  Under the condition in Eq.~\eqref{eq:block-full-multiplicity-condition}, every density matrix on $\mathcal H_\alpha$ satisfies Eq.~\eqref{eq:block-rank-condition}.  This proves Eq.~\eqref{eq:block-multiplicity-equality}.  Finally, $\mathfrak W(B)$ is the affine image of the state space of a finite direct-sum matrix algebra, so it is compact and convex, hence connected.
\end{proof}

\subsection{Commutative specialization: common spectral decompositions}

The first specialization has one-dimensional active blocks.  In that case, the active states carry no internal matrix data and are specified only by their probability weights over the blocks.

\begin{corollary}[Common spectral decompositions]
\label{cor:common-spectral-polytope}
Let $A=(A_1,\ldots,A_m)$ be a pairwise commuting Hermitian tuple, with joint spectral decomposition
\begin{equation}
A_i=\sum_{\alpha\in\Omega}a_i(\alpha)\Pi_\alpha,
\qquad a_i(\alpha)\in\R,
\end{equation}
where the $\Pi_\alpha$ are pairwise orthogonal joint spectral projectors and sum to $\id$.  If $\rank(\Pi_\alpha)\ge K$ for every $\alpha$, then
\begin{equation}
\Lambda_K(A)
=
\conv\{(a_1(\alpha),\ldots,a_m(\alpha)):\alpha\in\Omega\}.
\label{eq:common-spectral-polytope}
\end{equation}
\end{corollary}

\begin{proof}
Apply Theorem~\ref{thm:block-multiplicity} with $\mathcal H_\alpha=\C$ and multiplicity $m_\alpha=\rank(\Pi_\alpha)$.  The active state-space expectation body is precisely the displayed convex hull, and $m_\alpha\ge K\dim\mathcal H_\alpha=K$ gives equality.
\end{proof}

This corollary is the operator-theoretic source of the commuting Pauli polytopes in Sec.~\ref{sec:commuting}.  There the common eigenspaces have uniform multiplicity and a Clifford normal form makes both the sector values and the realizing projectors explicit.

\subsection{Single-block specialization: active/spectator reduction}
\label{subsec:reduction}

The second specialization has a single active block.  It separates what follows from normalized traces from what can be constructed using the spectator factor.

\begin{corollary}[Active/spectator reduction]
\label{thm:active-spectator}
Let $\mathcal H=\mathcal H_{\mathrm{act}}\otimes \mathcal H_{\mathrm{spect}}$ with $\dim\mathcal H_{\mathrm{spect}}=M$, and let
\begin{equation}
A_i = B_i\otimes \id_{M}\qquad (i=1,\ldots,m)
\label{eq:act-spect-form}
\end{equation}
for Hermitian $B_i$ on $\mathcal H_{\mathrm{act}}$. Fix $M\ge K$.  Every pure-state expectation tuple of the active operators is realized by a rank-$K$ scalar compression of the lifted tuple.  Indeed, if $\ket{\psi}\in\mathcal H_{\mathrm{act}}$ is a unit vector and $R$ is a rank-$K$ projector on $\mathcal H_{\mathrm{spect}}$, then
\begin{equation}
P_{\psi,R}:=\ket{\psi}\!\bra{\psi}\otimes R
\end{equation}
satisfies
\begin{equation}
P_{\psi,R}A_iP_{\psi,R}
=
\langle\psi|B_i|\psi\rangle\,P_{\psi,R},
\qquad i=1,\ldots,m.
\label{eq:active-product-compression}
\end{equation}
Conversely, every rank-$K$ scalar compression of the lifted tuple is determined, at the level of scalar values, by a density matrix on the active factor.  More precisely, if $P\in\cP_K(\mathcal H)$ and
\begin{equation}
PA_iP=\lambda_iP,\qquad i=1,\ldots,m,
\end{equation}
then, with
\begin{equation}
\rho_P:=\frac{P}{K},\qquad \sigma_P:=\Tr_{\mathrm{spect}}(\rho_P),
\end{equation}
one has
\begin{equation}
\lambda_i=\Tr(\sigma_PB_i),\qquad i=1,\ldots,m.
\label{eq:act-spect-reduced-state}
\end{equation}
Consequently,
\begin{equation}
\begin{aligned}
\Lambda_1(B_1,\ldots,B_m)
&\subseteq
\Lambda_K(A_1,\ldots,A_m)\\
&\subseteq
\conv\Lambda_1(B_1,\ldots,B_m).
\end{aligned}
\label{eq:active-spectator-sandwich}
\end{equation}
\end{corollary}

\begin{proof}
The product construction above proves the first inclusion in Eq.~\eqref{eq:active-spectator-sandwich}.  For the second, let $P\in\cP_K(\mathcal H)$ satisfy $PA_iP=\lambda_iP$.  Then
\begin{equation}
\lambda_i
=\frac{1}{K}\Tr(PA_i)
=\Tr\!\left[\rho_P(B_i\otimes\id_M)\right]
=\Tr(\sigma_PB_i).
\end{equation}
This proves Eq.~\eqref{eq:act-spect-reduced-state}.  The partial trace $\sigma_P$ is positive and has unit trace, so it is a density matrix on $\mathcal H_{\mathrm{act}}$.  Writing $\sigma_P=\sum_a q_a\ket{\phi_a}\!\bra{\phi_a}$ with $q_a\ge0$ and $\sum_aq_a=1$, we obtain
\begin{equation}
(\lambda_1,\ldots,\lambda_m)
=
\sum_aq_a
\bigl(\langle\phi_a|B_1|\phi_a\rangle,\ldots,
\langle\phi_a|B_m|\phi_a\rangle\bigr).
\end{equation}
The right-hand side is therefore a convex combination of vectors in $\Lambda_1(B_1,\ldots,B_m)$, which proves the second inclusion in Eq.~\eqref{eq:active-spectator-sandwich}.
\end{proof}

\begin{corollary}[Convex active ranges]
\label{cor:convex-active-range}
In the setting of Corollary~\ref{thm:active-spectator}, if $M\ge K$ and $\Lambda_1(B_1,\ldots,B_m)$ is already convex, then
\begin{equation}
\Lambda_K(B_1\otimes\id_M,\ldots,B_m\otimes\id_M)
=
\Lambda_1(B_1,\ldots,B_m).
\end{equation}
\end{corollary}

\begin{proof}
This is immediate from Eq.~\eqref{eq:active-spectator-sandwich}, since
\begin{equation}
\conv\Lambda_1(B_1,\ldots,B_m)=\Lambda_1(B_1,\ldots,B_m)
\end{equation}
under the stated convexity assumption.
\end{proof}

The reduced-state formula in Eq.~\eqref{eq:act-spect-reduced-state} shows that compression values are determined by the active marginal $\sigma_P$.  Set $D_{\mathrm{act}}:=\dim\mathcal H_{\mathrm{act}}$.  The uniform guarantee for realizing every pure-state expectation tuple begins at the threshold $M=K$.  Under the standing assumption $M\ge K$, there are two disjoint regimes.  If $K\le M<KD_{\mathrm{act}}$, every pure-state expectation tuple is realized, but full convexification is not guaranteed; the rank-two Pauli triple in Proposition~\ref{prop:rank2-pauli-ball} provides an explicit example in which it fails.  If $M\ge KD_{\mathrm{act}}$, Theorem~\ref{thm:block-multiplicity} realizes every active density matrix, and hence
\begin{equation}
\Lambda_K(B_1\otimes\id_M,\ldots,B_m\otimes\id_M)
=
\conv\Lambda_1(B_1,\ldots,B_m).
\label{eq:block-active-full-convexification}
\end{equation}


\section{Commuting Pauli geometry}
\label{sec:commuting}

Corollary~\ref{cor:common-spectral-polytope} identifies the rank-$K$ range of a tuple with one-dimensional active blocks as the convex hull of its common spectral data, provided each block multiplicity is at least $K$.  For commuting Pauli families, this abstract structure becomes explicit under a Clifford normal form: the generated abelian Pauli algebra is mapped to a $Z$-type algebra, and the joint eigenspace decomposition has one-dimensional active blocks with a uniform multiplicity $M$.  This leads to two rank regimes.  When $M\ge K$, the range is the full sector polytope, and the normal form yields explicit rank-$K$ projectors realizing every point.  When $M<K$, the convex-hull upper bound remains valid, but the full polytope need not be attained: the three-qubit local-$Z$ triple exhibits shrinking from a cube to an octahedron, collapse to the origin, and eventual emptiness.

\subsection{Commuting \texorpdfstring{$Z$}{Z}-type families}
\label{subsec:t0}

Fix $n\in\mathbb{N}$.  Let $Z_j$ denote the $n$-qubit Pauli operator that acts as $Z$ on qubit $j$ and as the identity on every other qubit.  We label phase-free Pauli-$Z$ strings by binary vectors $u=(u_1,\ldots,u_n)\in\mathbb F_2^n$ and write
\begin{equation}
Z^u
:=
\prod_{j=1}^n Z_j^{u_j}
=
Z^{u_1}\otimes\cdots\otimes Z^{u_n}.
\end{equation}
Let $S$ be a finite set of such strings.  Denote the $\mathbb F_2$-span of their labels by
\begin{equation}
\begin{aligned}
\mathcal{T}&:=\operatorname{span}_{\mathbb F_2}\{u\in\mathbb F_2^n:\ Z^u\in S\}\subseteq\mathbb F_2^n,\\
r&:=\dim_{\mathbb{F}_2}\mathcal{T},
\qquad
\ell:=n-r.
\end{aligned}
\end{equation}
We denote the uniform sector multiplicity in this section by
\begin{equation}
M:=2^\ell.
\end{equation}
Let $\langle S\rangle=\{Z^{u}:u\in\mathcal{T}\}$ be the abelian group generated by $S$.  The associated real group algebra is
\begin{equation}
\mathfrak{D}_S\ :=\ \operatorname{span}_{\R}\bigl\{Z^u : u\in\mathcal{T}\bigr\}.
\end{equation}
We say that a Hermitian operator is \emph{supported on $S$} if it lies in $\mathfrak D_S$.  Equivalently, every operator in $\mathfrak D_S$ is diagonal in the computational basis and depends only on the joint eigenvalue pattern of $\langle S\rangle$.  Indeed, for a computational basis vector $\ket{x}$, $x\in\mathbb F_2^n$, one has
\begin{equation}
Z^u\ket{x}=(-1)^{u\cdot x}\ket{x}.
\end{equation}
Thus two basis vectors $\ket{x}$ and $\ket{y}$ lie in the same joint eigenspace of $\langle S\rangle$ precisely when $(-1)^{u\cdot x}=(-1)^{u\cdot y}$ for every $u\in\mathcal T$.  For real coefficients $c_u$, the preceding identity gives
\begin{equation}
\left(\sum_{u\in\mathcal T}c_uZ^u\right)\ket{x}
=
\left(\sum_{u\in\mathcal T}c_u(-1)^{u\cdot x}\right)\ket{x},
\end{equation}
so the resulting eigenvalue is the same for all computational basis vectors in the same joint eigenspace.  Conversely, diagonal operators with this constancy property are exactly the real span of the characters $x\mapsto(-1)^{u\cdot x}$ with $u\in\mathcal T$.  Thus, we consider ordered Hermitian tuples $A=(A_1,\ldots,A_m)$ with each $A_i\in\mathfrak D_S$.

\subsection{Clifford reduction and sectorization}
\label{subsec:sectorization}

To make the common eigenspace structure of such tuples explicit, we first use a Clifford unitary to carry $\langle S\rangle$ to $\langle Z_1,\ldots,Z_r\rangle$.

\begin{lemma}[Clifford normal form]\label{lem:cliff}
There exists a Clifford unitary $U$ (generated by CNOTs and swaps) such that
\begin{equation}
U\,\langle S\rangle\,U^\dagger=\langle Z_1,\dots,Z_r\rangle.
\end{equation}
Consequently, for each $A_i\in\mathfrak{D}_S$ there is a real diagonal polynomial $f_i$ with
\begin{equation}
U A_i U^\dagger \;=\; f_i(Z_1,\dots,Z_r)\ \otimes\ \id_{2^\ell}.
\end{equation}
\end{lemma}

\begin{proof}
Identify the phase-free diagonal Pauli group with $\mathbb F_2^n$ through the injective map $u\mapsto Z^u$.  Under this identification, multiplication of $Z$-strings corresponds to addition of support vectors, and the generated group $\langle S\rangle$ is determined exactly by the subspace $\mathcal T$.

Choose an $\mathbb F_2$-basis $u^{(1)},\dots,u^{(r)}$ of $\mathcal T$ and let $G_S$ be the $r\times n$ generator matrix whose rows are these vectors.  Elementary row operations merely replace this basis by another basis of the same subspace, hence do not change $\mathcal T$ or $\langle S\rangle$.  After Gaussian elimination, followed by qubit swaps that move the pivot columns to the first $r$ positions, the transformed row space has a representative
\begin{equation}
G_S=[\,I_r\mid R\,],\qquad R\in\mathbb F_2^{r\times \ell}.
\end{equation}

The remaining block $R$ can be eliminated by CNOT gates.  Indeed, the conjugation rule
\begin{equation}
\mathrm{CNOT}_{c\to t}:Z_t\mapsto Z_cZ_t,\qquad Z_c\mapsto Z_c,
\end{equation}
acts on support vectors as the column operation $u_c\leftarrow u_c\oplus u_t$.  Hence, for every nonzero entry $R_{i,j}$, the gate $\mathrm{CNOT}_{(r+j)\to i}$ adds the pivot column $i$ to column $r+j$.  Since that pivot column is the standard basis vector $e_i$, this changes only the entry $R_{i,j}$ in the non-pivot column and turns it from $1$ to $0$.  Applying these gates to all nonzero entries of $R$ carries the row space to $\operatorname{span}_{\mathbb F_2}\{e_1,\dots,e_r\}$.

Let $U$ be the product of the swaps and CNOT gates realizing these column operations.  The induced action of $U$ on support vectors is an invertible linear map sending $\mathcal T$ onto $\operatorname{span}_{\mathbb F_2}\{e_1,\dots,e_r\}$.  Therefore
\begin{equation}
U\langle S\rangle U^\dagger = \langle Z_1,\dots,Z_r\rangle.
\end{equation}

For the second assertion, write $A_i=\sum_{u\in\mathcal T} c_{i,u}Z^u$ with real coefficients.  The construction above gives a bijection between $\mathcal T$ and $\mathbb F_2^r$: for each $u\in\mathcal T$ there is a unique $\alpha=(\alpha_1,\ldots,\alpha_r)\in\mathbb F_2^r$ such that
\begin{equation}
UZ^uU^\dagger=Z_1^{\alpha_1}\cdots Z_r^{\alpha_r}\otimes\id_{2^\ell}.
\end{equation}
Define $\widehat c_{i,\alpha}=c_{i,u}$ for this corresponding $u$. Then termwise conjugation gives
\begin{equation}
\begin{aligned}
UA_iU^\dagger
&=
\sum_{\alpha\in\mathbb F_2^r}\widehat c_{i,\alpha}
Z_1^{\alpha_1}\cdots Z_r^{\alpha_r}\otimes\id_{2^\ell}\\
&=
f_i(Z_1,\dots,Z_r)\otimes\id_{2^\ell},
\end{aligned}
\end{equation}
where
\begin{equation}
f_i(z_1,\dots,z_r)
:=
\sum_{\alpha\in\mathbb F_2^r}\widehat c_{i,\alpha}
z_1^{\alpha_1}\cdots z_r^{\alpha_r}.
\end{equation}
\end{proof}

Hence the Hilbert space factors as
\begin{equation}
\mathcal H\ \cong\ \mathcal H_{\mathrm{sect}}\otimes\mathcal H_{\mathrm{deg}}
\cong
(\C^2)^{\!\otimes r}\otimes(\C^2)^{\!\otimes \ell}.
\end{equation}
The degeneracy factor has dimension $\dim\mathcal H_{\mathrm{deg}}=M=2^\ell$.  We label the sectors by $\sigma=(\sigma_1,\dots,\sigma_r)\in\{\pm1\}^r$, where the components of $\sigma$ are the joint eigenvalues of $(Z_1,\dots,Z_r)$.  The corresponding sector projectors are
\begin{equation}
\Pi_\sigma = \ket{\sigma}\!\bra{\sigma}\otimes \id_M
\end{equation}
and the scalar eigenvalues $a_i(\sigma)\in\R$ are determined by $f_i(Z_1,\dots,Z_r)\ket{\sigma}=a_i(\sigma)\ket{\sigma}$.  In this basis, every $A_i$ has the spectral form
\begin{equation}
A_i\ =\ \sum_{\sigma\in\{\pm1\}^r} a_i(\sigma)\, \Pi_\sigma. \label{eq:sector-decomp}
\end{equation}

\subsection{Commuting-sector polytope and explicit projectors}
\label{subsec:polytope}

\begin{theorem}[Commuting-sector polytope]
\label{thm:polytope}
Let $A=(A_1,\dots,A_m)$ be a tuple of Hermitian operators in $\mathfrak{D}_S$, and let $M=2^{n-r}$ be the common sector multiplicity.  Write
\begin{equation}
\mathbf a(\sigma):=\bigl(a_1(\sigma),\ldots,a_m(\sigma)\bigr).
\end{equation}
Then, whenever $M\ge K$,
\begin{equation}
\Lambda_K(A)
=
\conv\bigl\{\mathbf a(\sigma):\ \sigma\in\{\pm1\}^r\bigr\}.
\label{eq:commuting-polytope}
\end{equation}
Moreover, in the Clifford coordinates of Lemma~\ref{lem:cliff}, every point of this convex hull has an explicit rank-$K$ realization.  For every probability vector $p=\{p_\sigma\}_{\sigma\in\{\pm1\}^r}$ with $p_\sigma\ge 0$ and $\sum_\sigma p_\sigma=1$, and every rank-$K$ projector $R_K$ on $\mathcal H_{\mathrm{deg}}$, define
\begin{equation}
P_{p,K}\ :=\ \ket{\psi_p}\!\bra{\psi_p}\otimes R_K,
\qquad \ket{\psi_p}:=\sum_{\sigma}\sqrt{p_\sigma}\,\ket{\sigma}.
\label{eq:explicit-commuting-projector}
\end{equation}
Then, for $i=1,\ldots,m$,
\begin{equation}
P_{p,K}(UA_iU^\dagger)P_{p,K}
=\lambda_i(p)P_{p,K},
\end{equation}
where
\begin{equation}
\lambda_i(p):=\sum_{\sigma}p_\sigma a_i(\sigma).
\end{equation}
Thus $U^\dagger P_{p,K}U$ is the corresponding realizing projector for the original tuple.  In particular, when $K=M$ one has $R_M=\id_M$, so
\begin{equation}
P_{p,M}=\ket{\psi_p}\!\bra{\psi_p}\otimes\id_M.
\end{equation}
\end{theorem}

\begin{proof}
By Lemma~\ref{lem:cliff}, the tuple has one-dimensional active blocks indexed by $\sigma$, each with multiplicity $M$.  Since $M\ge K$, Corollary~\ref{cor:common-spectral-polytope} gives the equality in Eq.~\eqref{eq:commuting-polytope}.  It remains only to record the particularly simple Pauli realization.  In the Clifford coordinates, Eq.~\eqref{eq:sector-decomp} gives
\begin{equation}
\bra{\psi_p}f_i(Z_1,\ldots,Z_r)\ket{\psi_p}
=
\sum_{\sigma}p_\sigma a_i(\sigma).
\end{equation}
Tensoring with $R_K$ gives
\begin{equation}
P_{p,K}(U A_i U^\dagger)P_{p,K}
=
\left(\sum_\sigma p_\sigma a_i(\sigma)\right)P_{p,K}.
\end{equation}
Thus $P_{p,K}$ realizes the convex combination determined by $p$; conjugating back by $U^\dagger$ gives the stated projector for the original tuple.
\end{proof}

\subsection{Consequences and a concrete example}
\label{subsec:consequences}

We next record a Carath\'eodory sparsity consequence and then present a three-qubit rank-two example, whose tetrahedral range will serve as the unrestricted baseline in Sec.~\ref{sec:collapse}.

\begin{proposition}[Carath\'eodory sparsity]
Under the hypotheses and notation of Theorem~\ref{thm:polytope}, if $M\ge K$, then any point of $\Lambda_K(A)\subset\R^m$ can be realized with a probability vector $p$ supported on at most $m{+}1$ sectors.
Equivalently, one may choose
\begin{equation}
\ket{\psi_p}=\sum_{j=1}^{q}\sqrt{w_j}\,\ket{\sigma^{(j)}},
\qquad
q\le m+1,
\end{equation}
for suitable sectors $\sigma^{(j)}$ and weights $w_j$.
\end{proposition}

\begin{proof}
By Theorem~\ref{thm:polytope}, every point of $\Lambda_K(A)$ lies in the convex hull of finitely many points in $\R^m$.  Carath\'eodory's theorem expresses it as a convex combination of at most $m+1$ of those sector points.  Using these weights in the projector $P_{p,K}$ from Eq.~\eqref{eq:explicit-commuting-projector} gives the claimed realization.
\end{proof}

As a concrete three-qubit illustration, take $n=3$ with $S=\{Z_1,Z_2\}$. Then $\mathcal T=\operatorname{span}_{\mathbb F_2}\{(1,0,0),(0,1,0)\}$, so $r=2$, $\ell=1$, and $M=2$. Let
\begin{equation}
E^{\mathrm{tet}}:=(Z_1,Z_2,Z_1Z_2).
\end{equation}
The Clifford normal form is already achieved: the first two qubits form the sector register, while the third qubit is the degeneracy register.  For $i,j\in\{0,1\}$, set
\begin{equation}
\Pi_{ij}:=\ket{ij}\!\bra{ij}\otimes\id_2.
\end{equation}
On the sector $\Pi_{ij}$ the three operators have joint eigenvalue
\begin{equation}
\mathbf e_{ij}:=\bigl((-1)^i,\,(-1)^j,\,(-1)^{i+j}\bigr).
\end{equation}
Thus the four sector values are
\begin{equation}
\begin{aligned}
\mathbf e_{00}&=(1,1,1),&
\mathbf e_{01}&=(1,-1,-1),\\
\mathbf e_{10}&=(-1,1,-1),&
\mathbf e_{11}&=(-1,-1,1),
\end{aligned}
\end{equation}
and Theorem~\ref{thm:polytope} gives
\begin{equation}
\Lambda_2(E^{\mathrm{tet}})
=
\conv\left\{
\begin{array}{l}
(1,1,1),\ (1,-1,-1),\\
(-1,1,-1),\ (-1,-1,1)
\end{array}
\right\}.
\end{equation}

Let us spell out the projector realizing an arbitrary point of this tetrahedron.  Given a target
$(x,y,z)$, define
\begin{equation}
\begin{aligned}
p_{00}&=\frac{1+x+y+z}{4}, & p_{01}&=\frac{1+x-y-z}{4},\\
p_{10}&=\frac{1-x+y-z}{4}, & p_{11}&=\frac{1-x-y+z}{4},
\end{aligned}
\end{equation}
so that
\begin{equation}
\begin{aligned}
\sum_{i,j}p_{ij}&=1,&
\sum_{i,j}p_{ij}(-1)^i&=x,\\
\sum_{i,j}p_{ij}(-1)^j&=y,&
\sum_{i,j}p_{ij}(-1)^{i+j}&=z.
\end{aligned}
\end{equation}
The condition that $(x,y,z)$ lies in the displayed tetrahedron is exactly the nonnegativity of these four barycentric coordinates.  Equivalently, the four facet inequalities are
\begin{equation}
\begin{aligned}
1+x+y+z&\ge0,&
1+x-y-z&\ge0,\\
1-x+y-z&\ge0,&
1-x-y+z&\ge0.
\end{aligned}
\end{equation}
When these inequalities hold, set
\begin{equation}
\ket{\psi_p}:=\sum_{i,j\in\{0,1\}}\sqrt{p_{ij}}\ket{ij},
\qquad
P_p:=\ket{\psi_p}\!\bra{\psi_p}\otimes\id_2.
\end{equation}
Since $\sum_{i,j}p_{ij}=1$, the vector $\ket{\psi_p}$ is normalized and $P_p$ is a rank-$2$ projector.  Moreover, the three entries of $E^{\mathrm{tet}}$ act diagonally on the sector register and trivially on the degeneracy register. Therefore
\begin{equation}
\begin{aligned}
P_pZ_1P_p&=\left(\sum_{i,j}p_{ij}(-1)^i\right)P_p=xP_p,\\
P_pZ_2P_p&=\left(\sum_{i,j}p_{ij}(-1)^j\right)P_p=yP_p,\\
P_p(Z_1Z_2)P_p&=\left(\sum_{i,j}p_{ij}(-1)^{i+j}\right)P_p=zP_p.
\end{aligned}
\end{equation}
Thus $(x,y,z)\in\Lambda_2(E^{\mathrm{tet}})$.  Conversely, if a rank-$2$ projector $P$ has compression vector $(\lambda_1,\lambda_2,\lambda_3)$ for $E^{\mathrm{tet}}$, then taking normalized traces in the common sector decomposition gives the sector weights
\begin{equation}
\mu_{ij}:=\frac{\Tr(P\Pi_{ij})}{\Tr(P)}
\end{equation}
which satisfy $\mu_{ij}\ge0$ and $\sum_{i,j}\mu_{ij}=1$ because $\sum_{i,j}\Pi_{ij}=\id$.  Hence they form a probability vector and satisfy
\begin{equation}
(\lambda_1,\lambda_2,\lambda_3)=\sum_{i,j}\mu_{ij}\mathbf e_{ij}.
\end{equation}
Hence no point outside the tetrahedron can occur.

\subsection{Beyond the sector-degeneracy threshold}
\label{subsec:rank-sensitive-summary}

Theorem~\ref{thm:polytope} identifies the full sector polytope only while the common eigensector multiplicity accommodates the compression rank, $M\ge K$.  Once $M<K$, the upper inclusion in Theorem~\ref{thm:block-multiplicity} still places the range inside the convex hull of the joint eigenvalue data, but that convex hull need not be attainable.  We call this the \emph{rank-over-sector-degeneracy} regime.  Unlike limited spectator multiplicity, which may retain every pure-state expectation tuple without filling its convex hull, this regime can shrink, collapse, or empty the range.

The full local-$Z$ Pauli triple on three qubits already exhibits all three outcomes.  Set
\begin{equation}
E^Z:=(Z_1,Z_2,Z_3).
\label{eq:three-Z-tuple}
\end{equation}
Its eight joint eigenvalue tuples are the vertices of the cube, and every common eigenspace is one-dimensional.  Thus Theorem~\ref{thm:polytope} gives the cube only at rank one; at higher ranks the exact ranges are as follows.

\begin{proposition}[Rank-sensitive three-qubit diagonal ranges]
\label{prop:three-Z-ranges}
For $E^Z$ in Eq.~\eqref{eq:three-Z-tuple},
\begin{equation}
\begin{aligned}
\Lambda_1(E^Z)
&=[-1,1]^3,\\
\Lambda_2(E^Z)
&=\{(x,y,z)\in\R^3:\ |x|+|y|+|z|\le1\},\\
\Lambda_3(E^Z)
&=\Lambda_4(E^Z)=\{(0,0,0)\},\\
\Lambda_K(E^Z)
&=\emptyset\qquad(K\ge5).
\end{aligned}
\label{eq:three-Z-ranges}
\end{equation}
\end{proposition}

The rank-sensitive transition of Proposition~\ref{prop:three-Z-ranges} is illustrated in Fig.~\ref{fig:rank-sensitive-ranges}: as $K$ increases, the cube first shrinks to an octahedron, then collapses to the origin, and finally becomes empty.  The contrast with the tetrahedral tuple above isolates the role of sector multiplicity.  That tuple acts trivially on the third qubit and has $M=2$, sufficient to realize its rank-two polytope, whereas $E^Z$ resolves all three computational-basis coordinates and has $M=1$, so $K=2$ already lies beyond the full-polytope regime.  General diagonal bounds and the proof of Proposition~\ref{prop:three-Z-ranges} are given in Appendix~\ref{app:rank-sensitive}.

\begin{figure*}[t]
  \centering
  \includegraphics[width=\textwidth]{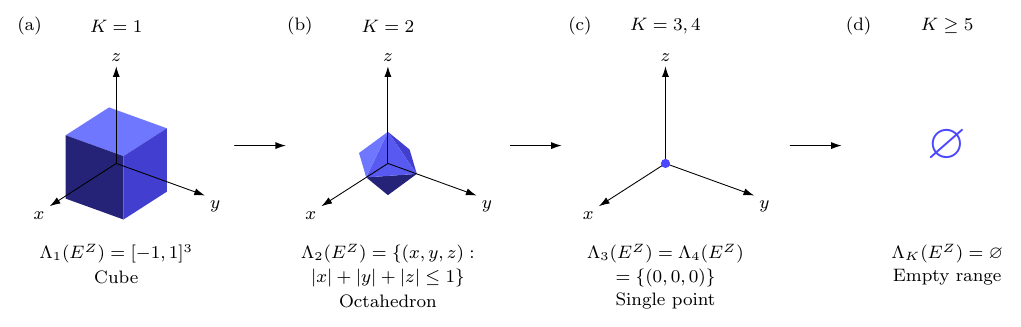}
  \caption{Rank-sensitive joint compression ranges for $E^Z=(Z_1,Z_2,Z_3)$, whose common eigensector multiplicity is $M=1$. At $K=M=1$ the range is the cube $[-1,1]^3$; once $K>M$, it is the octahedron $|x|+|y|+|z|\le1$ at rank two, collapses to the origin at ranks three and four, and is empty for $K\ge5$.}
  \label{fig:rank-sensitive-ranges}
\end{figure*}


\section{Subsystem codes and noncommuting geometry}
\label{sec:subsystem}

Corollary~\ref{thm:active-spectator} shows that, for a single active block with $D_{\mathrm{act}}=\dim\mathcal H_{\mathrm{act}}$, the condition $M\ge K$ realizes every pure-state expectation tuple, whereas $M\ge KD_{\mathrm{act}}$ realizes the full convex hull of these tuples.  Pauli subsystem codes supply this multiplicity algebraically: after a canonical Clifford reduction, the gauge algebra acts on the stabilizer and gauge degrees of freedom and trivially on the $2^k$-dimensional protected logical factor.  Thus, at the natural rank $K=2^k$, the logical subsystem is a base spectator space with $M=K$, and any unpinned stabilizer labels provide additional multiplicity.  The $[[4,1;1,2]]$ Bacon--Shor/compass block makes both thresholds explicit: stabilizer pinning changes $M$, producing a Bloch sphere when only pure-state realization is guaranteed and a filled Bloch ball once full convexification occurs; its lower-dimensional and mixed coordinate charts give disks, intervals, vertices, and product bodies.

\subsection{Canonical subsystem form}
\label{subsec:canonical}

To make the decomposition of a subsystem code explicit, we first use a Clifford unitary to separate its stabilizer, gauge, and protected logical degrees of freedom.

\begin{lemma}[Subsystem canonical form]
\label{lem:canonical}
Let \(\Gg\le\Pn\) be the operator gauge group of an \([[n,k;t,s]]\) Pauli subsystem code,
and let \(\Ss\) be the associated phase-fixed stabilizer subgroup from
Eq.~\eqref{eq:phase-fixed-stabilizer}.  There exists an \(n\)-qubit Clifford unitary \(U\) such that
\begin{equation}
\label{eq:canonical-form}
\begin{aligned}
U\,\Gg\,U^\dagger
&=
\Big\langle i\id,\ Z_{1},\dots,Z_{s},\
X_{s+1},Z_{s+1},\dots,X_{s+t},Z_{s+t}\Big\rangle,\\
U\Ss U^\dagger
&=\langle Z_1,\ldots,Z_s\rangle,
\end{aligned}
\end{equation}
and, consequently, for every $g\in \Gg$,
\begin{equation}
\label{eq:acts-trivially-on-logicals}
U g U^\dagger \;=\; O_g \otimes \id_{2^k},
\end{equation}
where $O_g$ acts on the first $s+t$ qubits.  In these coordinates the last $k$ qubits are the logical subsystem, and
\begin{equation}
U\Cc = \operatorname{span}\{\ket{0}^{\otimes s}\}\otimes(\C^2)^{\otimes t}\otimes(\C^2)^{\otimes k}.
\label{eq:canonical-codespace}
\end{equation}
\end{lemma}

\begin{proof}
We use the binary symplectic representation of the phase-free label space \(\Pnb\).  Thus the label \(\pi(P)\) of an \(n\)-qubit Pauli \(P\) is represented by a vector \(v_P=(\mathbf x\mid \mathbf z)\in\mathbb F_2^{2n}\), and commutation is encoded by
\begin{equation}
\omega(v_P,v_Q)
:=
\mathbf x_P\!\cdot\!\mathbf z_Q+\mathbf z_P\!\cdot\!\mathbf x_Q
\qquad (\operatorname{mod} 2).
\end{equation}
The standard vectors $e_j$ and $e_{n+j}$ represent $X_j$ and $Z_j$, respectively, and form the usual symplectic pairs.

Choose the paired Hermitian generating set in Eq.~\eqref{eq:gauge-generators-abstract} and write
\begin{equation}
\begin{alignedat}{2}
u_i&:=v_{S_i} &&\qquad (i=1,\ldots,s),\\
v_j&:=v_{X_{g_j}},\qquad
w_j:=v_{Z_{g_j}} &&\qquad (j=1,\ldots,t).
\end{alignedat}
\end{equation}
The defining commutation relations of a subsystem gauge group say exactly that
\begin{equation}
\omega(u_i,u_{i'})=\omega(u_i,v_j)=\omega(u_i,w_j)=0,
\end{equation}
and
\begin{equation}
\omega(v_j,w_{j'})=\delta_{jj'},\qquad
\omega(v_j,v_{j'})=\omega(w_j,w_{j'})=0.
\end{equation}
Equivalently, the symplectic Gram matrix of the ordered list
\begin{equation}
u_1,\ldots,u_s,\ v_1,w_1,\ldots,v_t,w_t
\end{equation}
is
\begin{equation}
\begin{pmatrix}0_s&0\\0&J_t\end{pmatrix},
\qquad
J_t=\bigoplus_{j=1}^t
\begin{pmatrix}0&1\\1&0\end{pmatrix}.
\end{equation}
The listed vectors are linearly independent because the chosen Pauli generators are independent modulo phase.

By the standard symplectic basis extension theorem over $\mathbb F_2$, this partial symplectic basis can be completed as follows.  There exist vectors
\begin{equation}
\tilde u_1,\ldots,\tilde u_s,\qquad a_1,b_1,\ldots,a_k,b_k,
\qquad k=n-s-t,
\end{equation}
such that the ordered list
\begin{equation}
\begin{gathered}
\tilde u_1,u_1,\ldots,\tilde u_s,u_s,\,
v_1,w_1,\ldots,v_t,w_t,\\
a_1,b_1,\ldots,a_k,b_k
\end{gathered}
\end{equation}
is a symplectic basis of $\mathbb F_2^{2n}$.  In particular, $\omega(\tilde u_i,u_{i'})=\delta_{ii'}$ and $\omega(a_\ell,b_{\ell'})=\delta_{\ell\ell'}$, with all pairings between distinct displayed pairs equal to zero.

Define a linear map $\phi:\mathbb F_2^{2n}\to\mathbb F_2^{2n}$ on this adapted basis by
\begin{equation}
\begin{alignedat}{3}
&\tilde u_i\mapsto e_i,\qquad &&u_i\mapsto e_{n+i}
  \qquad&&(i=1,\ldots,s),\\
&v_j\mapsto e_{s+j},\qquad &&w_j\mapsto e_{n+s+j}
  \qquad&&(j=1,\ldots,t),\\
&a_\ell\mapsto e_{s+t+\ell},\qquad
&&b_\ell\mapsto e_{n+s+t+\ell}
  \qquad&&(\ell=1,\ldots,k).
\end{alignedat}
\end{equation}
Because $\phi$ sends a symplectic basis to the standard symplectic basis, it lies in $\mathrm{Sp}(2n,\mathbb F_2)$.  The Clifford representation surjects onto this symplectic group, so there is an $n$-qubit Clifford unitary $V$ inducing $\phi$.  Hence
\begin{equation}
\begin{alignedat}{2}
VS_iV^\dagger&=\pm Z_i &&\qquad (i=1,\ldots,s),\\
VX_{g_j}V^\dagger&=\pm X_{s+j} &&\qquad (j=1,\ldots,t),\\
VZ_{g_j}V^\dagger&=\pm Z_{s+j} &&\qquad (j=1,\ldots,t).
\end{alignedat}
\end{equation}

The remaining signs can be removed without changing the symplectic action.  If the image of $S_i$ has sign $-1$, left-multiply $V$ by $X_i$, which flips $Z_i$ and commutes with all other canonical generators in the displayed list.  Similarly, left multiplication by $Z_{s+j}$ flips the sign of $X_{s+j}$ without changing $Z_{s+j}$, and left multiplication by $X_{s+j}$ flips the sign of $Z_{s+j}$ without changing $X_{s+j}$.  Applying these independent single-qubit Pauli corrections gives a Clifford unitary $U$ satisfying
\begin{equation}
US_iU^\dagger=Z_i,\qquad
UX_{g_j}U^\dagger=X_{s+j},\qquad
UZ_{g_j}U^\dagger=Z_{s+j}.
\end{equation}
Together with the scalar generator \(i\id\), which is fixed by conjugation, these operators generate \(U\Gg U^\dagger\).  This proves both identities in Eq.~\eqref{eq:canonical-form}.

In this canonical form every generator of $U\Gg U^\dagger$ acts trivially on qubits $s+t+1,\ldots,n$.  Therefore every product of such generators also acts trivially on those qubits.  Thus, for each $g\in\Gg$,
\begin{equation}
UgU^\dagger=O_g\otimes \id_{2^k},
\end{equation}
where $O_g$ acts on the first $s+t$ qubits.  This is Eq.~\eqref{eq:acts-trivially-on-logicals}.

Finally, the transformed phase-fixed stabilizer subgroup is \(\langle Z_1,\ldots,Z_s\rangle\); the full operator center is this subgroup times \(\Phin\). The common \(+1\) eigenspace of the phase-fixed subgroup is precisely
\begin{equation}
\operatorname{span}\{\ket{0}^{\otimes s}\}\otimes(\C^2)^{\otimes t}\otimes(\C^2)^{\otimes k},
\end{equation}
which is Eq.~\eqref{eq:canonical-codespace}.
\end{proof}

\subsection{Logical subsystems as multiplicity spaces}
\label{subsec:subsystem-multiplicity}

\begin{theorem}[Subsystem gauge algebras as multiplicity representations]
\label{thm:subsystem-multiplicity}
Let $\Gg\le\Pn$ be the operator gauge group of an $[[n,k;t,s]]$ Pauli subsystem code, and let $E=(E_1,\ldots,E_m)$ be a tuple of Hermitian Pauli observables from $\Gg$.  There exist an $n$-qubit Clifford unitary $U$ and Hermitian operators $B_1,\ldots,B_m$ on the first $s+t$ qubits such that
\begin{equation}
UE_iU^\dagger=B_i\otimes\id_{2^k},
\qquad i=1,\ldots,m.
\end{equation}
Consequently, at the natural compression rank $K=2^k$,
\begin{equation}
\Lambda_1(B_1,\ldots,B_m)
\subseteq
\Lambda_{2^k}(E)
\subseteq
\conv\Lambda_1(B_1,\ldots,B_m).
\label{eq:subsystem-active-sandwich}
\end{equation}
In particular, every pure-state expectation tuple of the active operators is realized by a scalar compression on a full protected logical block.
\end{theorem}

\begin{proof}
Lemma~\ref{lem:canonical} gives a Clifford unitary $U$ for which every gauge operator acts trivially on the last $k$ logical qubits, and hence gives the displayed factorization for each $E_i$.  Unitary invariance yields
\begin{equation}
\Lambda_{2^k}(E)=\Lambda_{2^k}(UEU^\dagger).
\end{equation}
Applying Corollary~\ref{thm:active-spectator} with spectator dimension $M=2^k$ and compression rank $K=2^k$ proves the two inclusions and the final assertion.
\end{proof}

Thus the protected logical factor is precisely a multiplicity space whose dimension equals the natural compression rank.  The analysis below focuses on the sector-restricted range $\Lambda_K^{\Pi}(A)$.  For an orthogonal projector $\Pi$, this is the joint rank-$K$ numerical range of the restricted tuple $(\Pi A_1\Pi,\ldots,\Pi A_m\Pi)$ acting on $\Pi\mathcal H_n$, or equivalently the subset of $\Lambda_K(A)$ obtained by requiring the realizing projector $P$ to satisfy $P\le\Pi$.

\subsection{Noncommuting geometry: the \texorpdfstring{$[[4,1;1,2]]$}{[[4,1;1,2]]} atlas}
\label{subsec:atlas412}

As a concrete example, we consider the $[[4,1;1,2]]$ Bacon--Shor/compass subsystem block and compute its rank-$2$ joint numerical ranges under different stabilizer pinnings.  We begin with its canonical form.  For notational orientation, we arrange the four qubits as
\begin{equation}
\begin{matrix}1&2\\3&4\end{matrix}
\label{eq:412-layout}
\end{equation}
and take the phase-fixed stabilizer generators to be
\begin{equation}
\begin{alignedat}{2}
S_X&:=X_1X_2X_3X_4,\qquad & S_Z&:=Z_1Z_2Z_3Z_4.
\end{alignedat}
\end{equation}
For the noncommuting gauge qubit, choose
\begin{equation}
\begin{alignedat}{3}
X_g&:=X_1X_2,\qquad       & Z_g&:=Z_2Z_4,\qquad & Y_g&:= i X_g Z_g.
\end{alignedat}
\end{equation}
Let \(\Gg=\langle i\id,S_X,S_Z,X_g,Z_g\rangle\) be the phaseful operator gauge group.  Its symplectic radical is spanned by the labels of \(S_X,S_Z\), and its phase-fixed stabilizer subgroup is \(\Ss=\langle S_X,S_Z\rangle\).  The gauge pair \(X_g,Z_g\) anticommutes, and the codespace \(\Cc=\{\ket{\psi}:S_X\ket{\psi}=\ket{\psi}=S_Z\ket{\psi}\}\) has \(\dim\Cc=2^{k+t}=4\) and factors as \(\Cc\simeq\mathcal H_G\otimes\mathcal H_L\) with \(\dim\mathcal H_L=\dim\mathcal H_G=2\). By Lemma~\ref{lem:canonical} (with \(n=4\), \(k=1\), \(t=1\), \(s=2\)) there exists a 4-qubit Clifford \(U\) with
\begin{equation}
\begin{aligned}
U\Gg U^\dagger
&=\langle i\id,Z_1,Z_2,X_3,Z_3\rangle,\\
UgU^\dagger
&=O_g\otimes\id_2,\qquad g\in\Gg,
\end{aligned}
\end{equation}
so that qubits $1$ and $2$ carry stabilizer labels, qubit $3$ is the gauge qubit, and qubit $4$ is the logical qubit.  In the rest of this subsection we work in these canonical coordinates.

Stabilizer pinning is modeled by the sector projector $\Pi_\zeta$ associated with a subset $\mathcal{Z}\subseteq\{Z_1,Z_2\}$ and a sign assignment \(\zeta:\mathcal Z\to\{\pm1\}\).  For \(Z_j\in\mathcal Z\), write \(\zeta_j:=\zeta(Z_j)\).  Then
\begin{equation}
\Pi_\zeta=\prod_{Z_j\in\mathcal{Z}}\frac{1}{2}\bigl(\id+\zeta_j Z_j\bigr),
\label{eq:stabilizer-sector-projector}
\end{equation}
and we write $q:=|\mathcal{Z}|\in\{0,1,2\}$. On the restricted Hilbert space $\Pi_\zeta\mathcal H_4$, the gauge qubit remains active and the unpinned stabilizer-label qubits together with the logical qubit form a spectator space of dimension
\begin{equation}
M=2^{3-q},
\label{eq:atlas-m}
\end{equation}
so the noncentral axes restrict to
\begin{equation}
X_3\otimes\id_M,\qquad Y_3\otimes\id_M,\qquad Z_3\otimes\id_M
\end{equation}
after identifying $\Pi_\zeta\mathcal H_4\cong\C^2_{\mathrm{gauge}}\otimes\C^M_{\mathrm{spect}}$.

\begin{proposition}[Rank-2 atlas for the four-qubit subsystem block]
\label{prop:412-atlas}
Fix a stabilizer sector $\Pi_\zeta$ as above, and let $q=|\mathcal Z|$. In canonical coordinates the full noncentral gauge triple has rank-$2$ range
\begin{equation}
\Lambda_2^{\Pi_\zeta}(X_3,Y_3,Z_3)
=
\begin{cases}
\mathbb S^2,& q=2,\\
\mathbb B^3,& q=0,1,
\end{cases}
\end{equation}
where $\mathbb S^2=\{(x,y,z):x^2+y^2+z^2=1\}$ and $\mathbb B^3=\{(x,y,z):x^2+y^2+z^2\le 1\}$.

Every two-axis gauge chart is the corresponding unit disk.  For example,
\begin{equation}
\Lambda_2^{\Pi_\zeta}(X_3,Z_3)
=
\{(x,z)\in\R^2:x^2+z^2\le 1\}
\end{equation}
for every $q$.  Likewise every one-axis gauge chart is the interval
\begin{equation}
\Lambda_2^{\Pi_\zeta}(X_3)=[-1,1]
\end{equation}
and the same statement holds after replacing $X_3$ by $Y_3$ or $Z_3$.

The central stabilizer-label coordinates retain their commuting polytope:
\begin{equation}
\Lambda_2^{\Pi_\zeta}(Z_1,Z_2)
=
C_\zeta(\mathcal Z).
\label{eq:central-atlas}
\end{equation}
Here
\begin{equation}
C_\zeta(\mathcal Z)
:=
[-1,1]^2\cap
\bigcap_{Z_j\in\mathcal Z}
\{\eta\in\R^2:\eta_j=\zeta_j\}.
\end{equation}
Thus no pinning gives the square, one pinned stabilizer gives an interval, and two pinned stabilizers give the vertex $(\zeta_1,\zeta_2)$.
\end{proposition}

\begin{proof}
In the canonical reduction above, the sector $\Pi_\zeta\mathcal H_4$ factors as $\C^2_{\mathrm{gauge}}\otimes\C^M_{\mathrm{spect}}$ with $M=2^{3-q}$.  The three gauge axes therefore have the form $(X\otimes\id_M,Y\otimes\id_M,Z\otimes\id_M)$ on this factorization.  Proposition~\ref{prop:rank2-pauli-ball} therefore gives the Bloch sphere when $q=2$, for which $M=2$, and the filled Bloch ball when $q=0$ or $1$, for which $M=8$ or $4$, respectively.  This sphere-to-ball transition is illustrated in Fig.~\ref{fig:412-atlas}.

It remains only to identify the lower-dimensional charts.  Any two distinct axes among $X_3,Y_3,Z_3$ are unitarily equivalent on the gauge qubit to $(X,Z)$.  Higher-rank numerical ranges are invariant under this conjugation, so Proposition~\ref{prop:rank2-pauli-lower-charts} gives the unit disk for every two-axis chart and the interval $[-1,1]$ for every one-axis chart.

For the central coordinates, each pinned $Z_j$ is scalar on $\Pi_\zeta\mathcal H_4$, so every compression value satisfies $\eta_j=\zeta_j$ whenever $Z_j\in\mathcal Z$; the norm bound $\lVert Z_j\rVert=1$ also gives $\eta_j\in[-1,1]$ for every unpinned coordinate.  This proves the inclusion into $C_\zeta(\mathcal Z)$.

For the reverse inclusion, fix $\eta\in C_\zeta(\mathcal Z)$.  On every pinned label qubit choose the $Z_j$-eigenvector with eigenvalue $\zeta_j$, and on every unpinned label qubit choose a unit vector $\ket{\phi_j}$ satisfying $\langle\phi_j|Z_j|\phi_j\rangle=\eta_j$.  Let $R$ be any rank-$2$ projector on the gauge--logical factor carried by qubits $3$ and $4$.  The tensor product of the chosen rank-one label projectors with $R$ is a rank-$2$ projector $P_\eta\le\Pi_\zeta$, and direct multiplication gives
\begin{equation}
P_\eta Z_j P_\eta=\eta_jP_\eta,
\qquad j=1,2.
\end{equation}
Thus every point of $C_\zeta(\mathcal Z)$ is attained, proving Eq.~\eqref{eq:central-atlas}.
\end{proof}

\begin{figure*}[t]
  \centering
  \includegraphics[width=\textwidth]{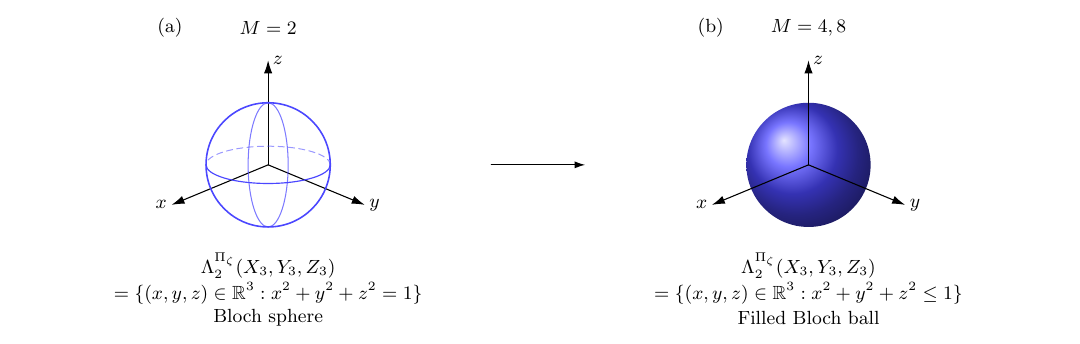}
  \caption{Rank-$2$ sphere-to-ball transition for the full noncentral gauge triple of the $[[4,1;1,2]]$ subsystem block, with $K=D_{\mathrm{act}}=2$ and $M=2^{3-q}$. (a) For $q=2$, one has $M=2=K<KD_{\mathrm{act}}=4$, and the range is the Bloch sphere. (b) For $q\in\{0,1\}$, one has $M\in\{4,8\}$ and hence $M\ge KD_{\mathrm{act}}=4$, and the range is the filled Bloch ball.}
  \label{fig:412-atlas}
\end{figure*}

The ranges obtained above fit the multiplicity framework of Theorem~\ref{thm:subsystem-multiplicity} and Corollary~\ref{thm:active-spectator}.  The protected logical qubit supplies the base spectator dimension $M=K=2$, while the active gauge qubit has dimension $D_{\mathrm{act}}=2$.  When both stabilizers are pinned, $M=2$ lies in the intermediate regime $K\le M<KD_{\mathrm{act}}$: every pure-state gauge expectation is realized, but full convexification is not guaranteed, consistently with the Bloch sphere found above.  If at least one stabilizer label remains unpinned, then $M=4$ or $8$, so $M\ge KD_{\mathrm{act}}$ and Eq.~\eqref{eq:block-active-full-convexification} yields the full Bloch ball.  For a two-axis chart, the pure-state range is already the convex unit disk, so Corollary~\ref{cor:convex-active-range} shows that the weaker condition $M\ge K$ suffices for every $q$.

The stabilizer and gauge coordinates may also be retained simultaneously, producing product bodies when sufficient spectator multiplicity remains.  A representative example, computed in Appendix~\ref{app:cylinder}, is
\begin{equation}
\Lambda_4(Z_1,X_3,Z_3)=\bigl\{(x,y,z)\in\R^3:\ |x|\le 1,\ y^2+z^2\le 1\bigr\},
\label{eq:stabilizer-gauge-cylinder}
\end{equation}
an interval$\times$disk cylinder.

Finally, on the codespace $(S_X,S_Z)=(+,+)$ every Hamiltonian built from the gauge algebra has the form
\begin{equation}
H_{\Cc}=H_{\mathrm{gauge}}\otimes\id_L.
\end{equation}
Hence the logical degeneracy is protected by the spectator factor.  If the gauge Hamiltonian has a unique ground state, the full ground space has rank $2$ and every gauge observable $g\in\Gg$ has a scalar compression on it:
\begin{equation}
P_{\mathrm{gs}}\,g\,P_{\mathrm{gs}}
=
\langle\phi_{\mathrm{gauge}}|O_g|\phi_{\mathrm{gauge}}\rangle\,P_{\mathrm{gs}},
\end{equation}
where $\ket{\phi_{\mathrm{gauge}}}$ is the active gauge ground state.  This is the subsystem-code version of Knill--Laflamme detectability: the gauge expectation may vary with the chosen gauge fixing, but it is invisible to the logical factor.


\section{Parent Hamiltonians}
\label{sec:hamiltonians}

Section~\ref{sec:subsystem} identified protected logical qubits algebraically as multiplicity spaces for the gauge algebra.  The same factorization has a dynamical interpretation: a Hamiltonian built from the active or gauge algebra can select an active ground state while acting trivially on the multiplicity factor.  Its ground projector can then realize scalar compressions of observables in the same algebra, while the multiplicity factor carries the residual ground-space degeneracy.  This mechanism also appears in stabilizer Hamiltonians and compass or bond-algebra models~\cite{BravyiTerhal2009,NussinovOrtiz2009}.

We first formulate a general template for active and spectator factors, including gauge fixing, and illustrate it with a minimal two-qubit path.  We then show that the spectator factor may itself be encoded rather than represented by an untouched physical qubit.  Finally, we return to the $[[4,1;1,2]]$ Bacon--Shor/compass block of Sec.~\ref{subsec:atlas412} and reinterpret its gauge geometry through the ground space of the compass Hamiltonian.

\subsection{General template: spectators, active blocks, and gauge fixing}
\label{subsec:ham-template}

Let a commuting family of Pauli symmetries decompose the Hilbert space into sectors
\begin{equation}
\mathcal H=\bigoplus_{\eta\in\Omega}
\mathcal H_{\mathrm{act},\eta}\otimes\mathcal H_{\mathrm{spect},\eta}.
\end{equation}

\begin{theorem}[Ground-space realization of scalar-compression vectors]
\label{thm:ground-space-compression}
Suppose a Hamiltonian $H(\theta)$ and a monitored Hermitian Pauli tuple $E=(E_1,\ldots,E_m)$ have the block forms
\begin{equation}
H(\theta)=\bigoplus_{\eta\in\Omega} H_\eta(\theta)\otimes\id_{M_\eta},
\qquad
E_j=\bigoplus_{\eta\in\Omega} B_{j,\eta}\otimes\id_{M_\eta},
\label{eq:sector-active-parent-form}
\end{equation}
where $M_\eta=\dim\mathcal H_{\mathrm{spect},\eta}$.  Assume that, for $\eta$ in a set of ground sectors \(\Omega_0(\theta)\), the active Hamiltonian \(H_\eta(\theta)\) has a unique normalized ground state \(\ket{\phi_\eta(\theta)}\), all these active ground energies are equal, and all other sectors have strictly higher ground energy.  Then the full ground-space projector is
\begin{equation}
P_\theta=\bigoplus_{\eta\in\Omega_0(\theta)}
\ket{\phi_\eta(\theta)}\!\bra{\phi_\eta(\theta)}\otimes\id_{M_\eta}.
\label{eq:sector-ground-projector}
\end{equation}
If, in addition, the active expectations
\begin{equation}
\lambda_j(\theta):=
\bra{\phi_\eta(\theta)}B_{j,\eta}\ket{\phi_\eta(\theta)}
\label{eq:sector-independent-expectation}
\end{equation}
are independent of \(\eta\in\Omega_0(\theta)\), then
\begin{equation}
P_\theta E_jP_\theta=\lambda_j(\theta)P_\theta,\qquad j=1,\ldots,m.
\label{eq:ground-projector-compression}
\end{equation}
Thus the ground space realizes the point \(\lambda(\theta)=(\lambda_1(\theta),\ldots,\lambda_m(\theta))\in\Lambda_{K(\theta)}(E_1,\ldots,E_m)\), where
\begin{equation}
K(\theta)=\rank P_\theta=\sum_{\eta\in\Omega_0(\theta)}M_\eta.
\end{equation}
\end{theorem}

\begin{proof}
The strict separation of the sector ground energies implies that the full ground space is the direct sum, over $\eta\in\Omega_0(\theta)$, of the active ground-state line tensored with the entire spectator factor.  This gives the formula for $P_\theta$ and its rank.  Writing
\begin{equation}
c_{j,\eta}(\theta)
:=
\bra{\phi_\eta(\theta)}B_{j,\eta}\ket{\phi_\eta(\theta)},
\end{equation}
one has
\begin{equation}
P_\theta E_jP_\theta
=
\bigoplus_{\eta\in\Omega_0(\theta)}
c_{j,\eta}(\theta)
\ket{\phi_\eta(\theta)}\!\bra{\phi_\eta(\theta)}\otimes\id_{M_\eta}.
\end{equation}
When the expectations are sector-independent, $c_{j,\eta}(\theta)=\lambda_j(\theta)$ throughout the direct sum, which proves the claim.
\end{proof}

A particularly common situation occurs when the Hamiltonian and monitored observables act trivially on a common spectator factor:
\begin{equation}
H(\theta)=H_{\mathrm{act}}(\theta)\otimes\id_M,\qquad
E_j=B_j\otimes\id_M.
\label{eq:product-parent-form}
\end{equation}
If \(H_{\mathrm{act}}(\theta)\) has a unique ground vector \(\ket{\phi_\theta}\), then
\begin{equation}
P_\theta=\ket{\phi_\theta}\!\bra{\phi_\theta}\otimes\id_M
\end{equation}
has rank \(M\) and satisfies
\begin{equation}
P_\theta E_jP_\theta=
\bra{\phi_\theta}B_j\ket{\phi_\theta}\,P_\theta.
\label{eq:product-parent-compression}
\end{equation}
This is precisely the construction of Corollary~\ref{thm:active-spectator}, now viewed through Theorem~\ref{thm:ground-space-compression} as a ground-space construction.

Gauge fixing is implemented by adding commuting Pauli penalty terms.  Let \(C_1,\ldots,C_p\) be mutually commuting Hermitian Pauli observables satisfying \(C_a^2=\id\) and \([C_a,H(\theta)]=0\), and let \(\xi=(\xi_1,\ldots,\xi_p)\in\{\pm1\}^p\) be a compatible sign assignment, meaning that the joint eigenspace
\begin{equation}
\mathcal H_\xi
:=
\{\ket{\psi}:C_a\ket{\psi}=\xi_a\ket{\psi}\ \text{for all }a\}
\end{equation}
is nonzero.  Adding
\begin{equation}
H_{\mathrm{fix}}=-\Delta\sum_{a=1}^p \xi_a C_a,
\qquad \Delta>0,
\end{equation}
shifts a nonzero joint sector with eigenvalue pattern \(\eta=(\eta_1,\ldots,\eta_p)\) by \(-\Delta\sum_a\xi_a\eta_a\).  Hence, when \(\Delta\) is sufficiently large compared with the differences among the sector ground energies, the ground space of \(H(\theta)+H_{\mathrm{fix}}\) is contained in \(\mathcal H_\xi\), and no other joint sector contributes.  On \(\mathcal H_\xi\), the penalty is the scalar \(-p\Delta\), so the restriction of the active Hamiltonian is unchanged up to this additive constant.

In the exact product case in Eq.~\eqref{eq:product-parent-form}, if the \(C_a\) act only on the spectator factor, all spectator sectors have the same unperturbed active ground energy and every \(\Delta>0\) suffices.  Writing \(\Pi_\xi\) for the joint spectator-sector projector, the gauge-fixed ground projector is
\begin{equation}
P_{\theta,\xi}
=
\ket{\phi_\theta}\!\bra{\phi_\theta}\otimes\Pi_\xi,
\qquad
\rank P_{\theta,\xi}=\rank\Pi_\xi.
\end{equation}

As a minimal illustration of the sector construction and its gauge fixing, consider the two-qubit path
\begin{equation}
H(\theta)=-(1-\theta)\,Z\!\otimes\!Z-\theta\,X\!\otimes\!I,
\qquad \theta\in[0,1],
\end{equation}
and write \(S=I\otimes Z\).  Since \([H(\theta),S]=0\), the two eigenspaces \(S=\zeta\), \(\zeta=\pm1\), are superselection sectors.  In the sector \(S=\zeta\),
\begin{equation}
H_\zeta(\theta)=-(1-\theta)\zeta\,Z_1-\theta X_1.
\end{equation}
Let
\begin{equation}
\varepsilon(\theta):=\sqrt{(1-\theta)^2+\theta^2}.
\end{equation}
For every \(\theta\), each \(H_\zeta(\theta)\) has a unique ground vector \(\ket{\phi_\zeta(\theta)}\) with ground energy \(-\varepsilon(\theta)\).  Thus the full ground projector is
\begin{equation}
P_\theta=\sum_{\zeta=\pm1}
\ket{\phi_\zeta(\theta)}\!\bra{\phi_\zeta(\theta)}
\otimes
\frac{\id+\zeta Z}{2},
\qquad \rank P_\theta=2.
\end{equation}
The two monitored Pauli observables have sector-independent ground expectations:
\begin{equation}
\begin{aligned}
\bra{\phi_\zeta(\theta)}\zeta Z_1\ket{\phi_\zeta(\theta)}
&=\frac{1-\theta}{\varepsilon(\theta)},\\
\bra{\phi_\zeta(\theta)}X_1\ket{\phi_\zeta(\theta)}
&=\frac{\theta}{\varepsilon(\theta)}.
\end{aligned}
\end{equation}
Consequently,
\begin{equation}
\begin{aligned}
P_\theta(Z\otimes Z)P_\theta&=\frac{1-\theta}{\varepsilon(\theta)}\,P_\theta,\\
P_\theta(X\otimes I)P_\theta&=\frac{\theta}{\varepsilon(\theta)}\,P_\theta.
\end{aligned}
\label{eq:n2-compressions}
\end{equation}
Writing
\begin{equation}
\lambda_{ZZ}(\theta)
:=
\frac{1-\theta}{\varepsilon(\theta)},
\qquad
\lambda_{XI}(\theta)
:=
\frac{\theta}{\varepsilon(\theta)}
\label{eq:n2-compression-coordinates}
\end{equation}
for the two compression coordinates, we obtain
\begin{equation}
\bigl(\lambda_{ZZ}(\theta),\lambda_{XI}(\theta)\bigr)
\in\Lambda_2(Z\otimes Z,X\otimes I),
\label{eq:n2-compression-range}
\end{equation}
where
\begin{equation}
\lambda_{ZZ}(\theta)^2+\lambda_{XI}(\theta)^2=1.
\label{eq:n2-quarter-circle}
\end{equation}
Thus, as $\theta$ varies from $0$ to $1$, this compression point traces the first-quadrant unit-circle arc shown in Fig.~\ref{fig:parent-hamiltonian-path}.  Adding a large term \(-\Delta\,\zeta_0(I\otimes Z)\) selects the sector \(S=\zeta_0\) and turns the same active interpolation into a nondegenerate rank-one ground space.

\begin{figure}[t]
  \centering
  \includegraphics[width=0.82\columnwidth]{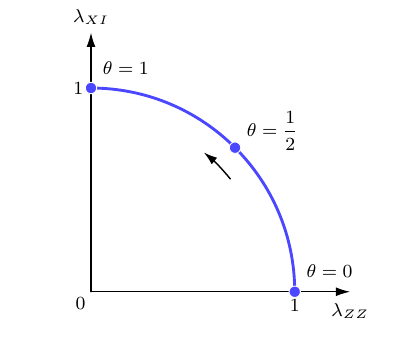}
  \caption{Parent-Hamiltonian interpolation for the two-qubit example. The rank-two ground projector traces the first-quadrant unit-circle arc in the $(\lambda_{ZZ},\lambda_{XI})$ plane as $\theta$ varies from $0$ to $1$. The arrow indicates increasing $\theta$.}
  \label{fig:parent-hamiltonian-path}
\end{figure}

\subsection{Encoded spectators and interpolation paths}
\label{ex:n4-r3}
A spectator factor need not coincide with an untouched physical qubit in the displayed Hamiltonian; it may instead be encoded in the commutant.  As a representative example, let
\begin{equation}
H(\theta)=-(1-\theta)(Z_1Z_2+Z_2Z_3+Z_3Z_4)-\theta(X_1+X_2+X_3).
\end{equation}
Although \(Z_3Z_4\) contains qubit~4, this qubit is an encoded spectator rather than an active degree of freedom.  Indeed \(Z_4\) and \(X_1X_2X_3X_4\) commute with \(H(\theta)\) for every \(\theta\) and anticommute with each other, so they generate a Pauli algebra inside the commutant.  The Clifford
\begin{equation}
U=\mathrm{CNOT}_{4\to 1}\,\mathrm{CNOT}_{4\to 2}\,\mathrm{CNOT}_{4\to 3}
\end{equation}
identifies this encoded spectator with qubit~4.  Using \(\mathrm{CNOT}_{c\to t}: Z_t\mapsto Z_cZ_t\), \(Z_c\mapsto Z_c\), and \(X_c\mapsto X_cX_t\), \(X_t\mapsto X_t\), the three gates send \(Z_3Z_4\) to \(Z_3\): the last CNOT contributes an extra \(Z_4\) to \(Z_3\), which cancels the original \(Z_4\).  The same rules show that the extra \(Z_4\) factors cancel in \(Z_1Z_2\) and \(Z_2Z_3\), while the target-qubit operators \(X_1,X_2,X_3\) are unchanged.  Hence
\begin{equation}
U H(\theta) U^\dagger = H_{\mathrm{act}}(\theta)\otimes\id_{\mathrm{spect}}^{(4)},
\end{equation}
where
\begin{equation}
H_{\mathrm{act}}(\theta):=-(1-\theta)(Z_1Z_2+Z_2Z_3+Z_3)-\theta(X_1+X_2+X_3).
\end{equation}
The active Hamiltonian has a unique ground state for all \(\theta\in[0,1]\).  At \(\theta=0\) the three commuting constraints \(Z_1Z_2=Z_2Z_3=Z_3=+1\) have a unique common solution.  For \(0<\theta\le1\), the computational-basis matrix is stoquastic and the \(X_1,X_2,X_3\) flips connect the active basis graph; Perron--Frobenius gives a unique ground vector \(\ket{\phi_\theta}\).  Therefore
\begin{equation}
U P_\theta U^\dagger=\ket{\phi_\theta}\!\bra{\phi_\theta}\otimes\id_2,
\qquad \rank P_\theta=2.
\end{equation}
For the tuple
\begin{equation}
E=(Z_1Z_2,\ Z_2Z_3,\ Z_3Z_4,\ X_1,\ X_2,\ X_3),
\end{equation}
the compression vector is
\begin{equation}
\begin{aligned}
\lambda(\theta)=\bigl(&
\langle Z_1Z_2\rangle_{\phi_\theta},
\langle Z_2Z_3\rangle_{\phi_\theta},
\langle Z_3\rangle_{\phi_\theta},\\
&
\langle X_1\rangle_{\phi_\theta},
\langle X_2\rangle_{\phi_\theta},
\langle X_3\rangle_{\phi_\theta}
\bigr),
\end{aligned}
\end{equation}
after the displayed Clifford reduction.  Thus \(P_\theta\) realizes a rank-\(2\) curve in \(\Lambda_2(E)\).  Adding a large \(-\Delta Z_4\) term selects one encoded-spectator eigenstate without changing the active spectrum.

\subsection{Mini Bacon--Shor / compass on a \texorpdfstring{$2\times 2$}{2 x 2} lattice}
\label{ex:bacon-shor}
On the $2\times 2$ layout in Eq.~\eqref{eq:412-layout}, take the row/column gauge terms
\begin{equation}
X_1X_2,\qquad X_3X_4,\qquad Z_1Z_3,\qquad Z_2Z_4,
\end{equation}
with stabilizers
\begin{equation}
S_X=X_1X_2X_3X_4,\qquad S_Z=Z_1Z_2Z_3Z_4,
\end{equation}
and the compass Hamiltonian with stabilizer penalties
\begin{equation}
\begin{aligned}
H_{\mathrm{comp}}(J_x,J_z,\Delta)
=&-J_x(X_1X_2+X_3X_4)\\
&-J_z(Z_1Z_3+Z_2Z_4)
-\Delta(S_X+S_Z).
\end{aligned}
\end{equation}
As in Sec.~\ref{subsec:atlas412}, we choose \(X_g=X_1X_2\) and \(Z_g=Z_2Z_4\).  In a stabilizer sector \((S_X,S_Z)=(\xi_X,\xi_Z)\), the row and column sums reduce to
\begin{equation}
\begin{aligned}
X_1X_2+X_3X_4&=(1+\xi_X)X_g,\\
Z_1Z_3+Z_2Z_4&=(1+\xi_Z)Z_g,
\end{aligned}
\end{equation}
so, after the canonical Clifford of Lemma~\ref{lem:canonical}, the sector Hamiltonian is
\begin{equation}
\begin{aligned}
H_{\xi_X,\xi_Z}
=&-\bigl[J_x(1+\xi_X)X_3
       +J_z(1+\xi_Z)Z_3\bigr]\otimes\id_2\\
&-\Delta(\xi_X+\xi_Z)\id.
\end{aligned}
\label{eq:compass-sector-hamiltonian}
\end{equation}
The identity factor is the protected logical qubit.  For \(J_x,J_z>0\) and \(\Delta\ge0\), the \((\xi_X,\xi_Z)=(+,+)\) sector is the unique lowest stabilizer sector.  Inside it the active Hamiltonian is
\begin{equation}
-2(J_xX_3+J_zZ_3),
\end{equation}
whose ground state is unique unless \(J_x=J_z=0\).  Therefore the physical ground projector has rank \(2\) and, with \(R=\sqrt{J_x^2+J_z^2}\),
\begin{equation}
\begin{aligned}
P_{\mathrm{gs}}X_gP_{\mathrm{gs}}&=\frac{J_x}{R}P_{\mathrm{gs}},\\
P_{\mathrm{gs}}Z_gP_{\mathrm{gs}}&=\frac{J_z}{R}P_{\mathrm{gs}},\\
P_{\mathrm{gs}}Y_gP_{\mathrm{gs}}&=0.
\end{aligned}
\label{eq:compass-gauge-compression}
\end{equation}
Thus the compass parent Hamiltonian selects a point on the Bloch sphere of the pinned gauge qubit, while the rank-two logical degeneracy is untouched.

This example also connects directly to the atlas of Proposition~\ref{prop:412-atlas}.  If both stabilizers are pinned, the full gauge triple \((X_g,Y_g,Z_g)\) has spectator dimension \(M=2\) and rank-\(2\) range \(\mathbb S^2\).  If one studies the same gauge triple before pinning all stabilizer labels, the additional sector multiplicity supplies \(M=4\) or \(8\), and the corresponding rank-\(2\) range is the filled Bloch ball \(\mathbb B^3\).  The parent Hamiltonian chooses a ground-state point in these bodies, while the joint numerical range describes the whole family of gauge fixings compatible with the same subsystem structure.

Additional tensor-product, cluster, and mixed-Pauli examples are collected in Appendix~\ref{app:parent-hamiltonian-examples}.  Together, Secs.~\ref{sec:subsystem} and~\ref{sec:hamiltonians} provide the algebraic and ground-space realizations of the subsystem multiplicity mechanism.


\section{A disconnected signature spectrum example}
\label{sec:collapse}

The example below exhibits a complementary obstruction to the multiplicity
mechanisms developed above.  Beginning with the rank-two tetrahedral
range of Sec.~\ref{subsec:consequences}, we impose additional noncommuting Pauli detection constraints. They do not change which sector weights the trace condition allows, but they forbid the superpositions that would realize most of those weights, leaving only the barycenter and four vertices in the monitored range.

\subsection{The monitored tetrahedron and the enlarged detection family}

In the absence of any additional compression equations, the commuting tuple from Sec.~\ref{subsec:consequences} is
\begin{equation}
E^{\mathrm{tet}}=(Z_1,Z_2,Z_1Z_2).
\label{eq:collapse-monitored-triple}
\end{equation}
Its rank-two range is the tetrahedron obtained there.  We now place the same monitored coordinates inside a larger, noncommuting Pauli detection problem.  The resulting full tuple is not a commuting family; the point of this section is precisely to determine how additional noncommuting constraints alter the commuting tetrahedron.

On the first two qubits define the seven-dimensional Pauli subspace
\begin{equation}
\mathfrak B
:=
\operatorname{span}_{\C}
\{II,XI,IX,XX,ZI,IZ,ZZ\}
\subset M_4(\C).
\label{eq:collapse-B-space}
\end{equation}
Its Hilbert--Schmidt orthogonal complement is $\mathfrak B^\perp=\operatorname{span}_{\C}\mathcal R$, where
\begin{equation}
\mathcal R:=\{YI,IY,YY,XY,YX,XZ,ZX,YZ,ZY\}.
\label{eq:collapse-R-set}
\end{equation}
Define
\begin{equation}
\mathcal F^{\mathrm{bg}}
:=
\{R\otimes Q:R\in\mathcal R,\ Q\in\{I,X,Y,Z\}\},
\label{eq:collapse-background-family}
\end{equation}
so $|\mathcal F^{\mathrm{bg}}|=36$.  Fix any ordering
$E^{\mathrm{bg}}$ of $\mathcal F^{\mathrm{bg}}$.  Appending the monitored
triple $E^{\mathrm{tet}}$, write
\begin{equation}
E^{\mathrm{full}}:=(E^{\mathrm{bg}},E^{\mathrm{tet}})
\label{eq:collapse-full-tuple}
\end{equation}
for the resulting $39$-entry Pauli tuple.

\subsection{Classification of rank-one elements}

The feasible projectors will be classified by two sets of normalized product states on the first two qubits:
\begin{align}
\mathcal W_Z
&:=\{\ket{00},\ket{01},\ket{10},\ket{11}\},
\label{eq:collapse-WZ}\\
\mathcal W_X
&:=\{\ket{++},\ket{+-},\ket{-+},\ket{--}\}.
\label{eq:collapse-WX}
\end{align}
The algebraic input is the following classification of the rank-one elements of $\mathfrak B$.

\begin{lemma}[Rank-one classification]
\label{lem:collapse-rank-one-classification}
Let $u,v\in\C^4$ be nonzero.  If
\begin{equation}
uv^\dagger\in\mathfrak B,
\end{equation}
then there exists $w\in\mathcal W_Z\cup\mathcal W_X$ such that $u$ and $v$ are both proportional to $w$.
\end{lemma}

\begin{proof}
In the computational basis $\ket{00},\ket{01},\ket{10},\ket{11}$, every element of $\mathfrak B$ has the form
\begin{equation}
\begin{pmatrix}
r_1&a&b&c\\
a&r_2&c&b\\
b&c&r_3&a\\
c&b&a&r_4
\end{pmatrix}.
\label{eq:collapse-B-matrix-form}
\end{equation}
In particular, every matrix in $\mathfrak B$ is symmetric.
Let $u,v\in\C^4$ be nonzero and suppose $uv^\dagger\in\mathfrak B$.  Symmetry of Eq.~\eqref{eq:collapse-B-matrix-form} gives
\begin{equation}
uv^\dagger=(uv^\dagger)^T=\overline v\,u^T.
\end{equation}
Both sides have rank one, so their column spaces coincide.  Hence $v\propto\overline u$, and, up to a nonzero scalar, the matrix $uv^\dagger$ equals $uu^T$.

Write $u=(u_1,u_2,u_3,u_4)^T$.  Comparing the three pairs of repeated off-diagonal entries in Eq.~\eqref{eq:collapse-B-matrix-form} yields
\begin{equation}
u_1u_2=u_3u_4,
\qquad
u_1u_3=u_2u_4,
\qquad
u_1u_4=u_2u_3.
\label{eq:collapse-rank-one-equations}
\end{equation}
If one coordinate vanishes, these equations force all but one coordinate to vanish.  Thus $u$ is proportional to one of the four computational-basis vectors in $\mathcal W_Z$.

Suppose instead that no coordinate vanishes.  Scale so that $u_1=1$.  Equation~\eqref{eq:collapse-rank-one-equations} becomes
\begin{equation}
u_2=u_3u_4,
\qquad
u_3=u_2u_4,
\qquad
u_4=u_2u_3.
\end{equation}
They imply $u_2^2=u_3^2=u_4^2=1$, with the signs constrained by the first equation.  Hence
\begin{equation}
u=(1,\epsilon,\eta,\epsilon\eta)^T,
\qquad \epsilon,\eta\in\{+1,-1\},
\end{equation}
which is proportional to one of the four product $X$-basis states in $\mathcal W_X$.  Since $v\propto\overline u$ and all these representatives are real, $u$ and $v$ are proportional to the same $w\in\mathcal W_Z\cup\mathcal W_X$.
\end{proof}

\subsection{Five-point range under Pauli constraints}

\begin{theorem}[Exact five-point Pauli compression range]
\label{thm:five-point-range}
A rank-two projector $P$ detects every entry of $E^{\mathrm{full}}$ if and only if
\begin{equation}
P=P_w:=\ket{w}\!\bra{w}\otimes\id_2
\label{eq:collapse-Pw}
\end{equation}
for some $w\in\mathcal W_Z\cup\mathcal W_X$.  Define
\begin{equation}
\begin{aligned}
\mathbf e_{00}&=(1,1,1),&
\mathbf e_{01}&=(1,-1,-1),\\
\mathbf e_{10}&=(-1,1,-1),&
\mathbf e_{11}&=(-1,-1,1),
\end{aligned}
\label{eq:collapse-tetrahedron-vertices}
\end{equation}
and
\begin{equation}
V_5:=\{\mathbf0_3,\mathbf e_{00},\mathbf e_{01},\mathbf e_{10},\mathbf e_{11}\}\subset\R^3.
\label{eq:collapse-V5}
\end{equation}
Then the ordinary joint higher-rank numerical range of the full tuple is
\begin{equation}
\Lambda_2(E^{\mathrm{full}})
=
\{(\mathbf0_{36},\mathbf v):\mathbf v\in V_5\}.
\label{eq:collapse-full-range}
\end{equation}
Equivalently, the monitored triple under the $36$ additional background constraints has range
\begin{equation}
\Lambda_2(E^{\mathrm{tet}}\mid\mathcal F^{\mathrm{bg}})
=V_5.
\label{eq:collapse-monitored-range}
\end{equation}
\end{theorem}

\begin{proof}
Suppose first that $P$ detects every entry of $E^{\mathrm{full}}$, and choose an orthonormal basis $\ket{0_L},\ket{1_L}$ of its range.  Define the off-diagonal logical matrix unit
\begin{equation}
L_{01}:=\ket{0_L}\!\bra{1_L}.
\label{eq:collapse-logical-matrix-unit}
\end{equation}
For every $E\in\mathcal F^{\mathrm{bg}}$, scalar compression implies
\begin{equation}
\bra{1_L}E\ket{0_L}=0,
\end{equation}
or equivalently
\begin{equation}
\Tr(E^\dagger L_{01})=0.
\end{equation}
The four single-qubit Pauli matrices span $M_2(\C)$, while $\mathcal R$ spans $\mathfrak B^\perp$.  Therefore
\begin{equation}
\operatorname{span}_{\C}(\mathcal F^{\mathrm{bg}})
=
\mathfrak B^\perp\otimes M_2(\C),
\end{equation}
and Hilbert--Schmidt orthogonality gives
\begin{equation}
L_{01}
\in
\left(\mathfrak B^\perp\otimes M_2(\C)\right)^\perp
=
\mathfrak B\otimes M_2(\C).
\label{eq:collapse-L01-in-B}
\end{equation}

Decompose the logical vectors with respect to the third qubit:
\begin{equation}
\begin{aligned}
\ket{0_L}&=u_0\otimes\ket0+u_1\otimes\ket1,\\
\ket{1_L}&=v_0\otimes\ket0+v_1\otimes\ket1,
\end{aligned}
\label{eq:collapse-logical-blocks}
\end{equation}
where $u_0,u_1,v_0,v_1\in\C^4$.  The four third-qubit blocks of $L_{01}$ are
\begin{equation}
u_av_b^\dagger,
\qquad a,b\in\{0,1\}.
\end{equation}
The inclusion in Eq.~\eqref{eq:collapse-L01-in-B} implies that every nonzero such block belongs to $\mathfrak B$.

Choose one nonzero $u_{a_0}$ and one nonzero $v_{b_0}$.  Lemma~\ref{lem:collapse-rank-one-classification} shows that both are proportional to a common $w\in\mathcal W_Z\cup\mathcal W_X$.  Pairing every other nonzero $u_a$ with $v_{b_0}$ and every other nonzero $v_b$ with $u_{a_0}$ forces the same $w$ for every block.  Consequently, there exist $\alpha,\beta\in\C^2$ such that
\begin{equation}
\ket{0_L}=w\otimes\alpha,
\qquad
\ket{1_L}=w\otimes\beta.
\label{eq:collapse-logical-factorization}
\end{equation}
Because $w$ is normalized, normalization and orthogonality of the logical states imply that $\alpha,\beta$ form an orthonormal basis of the third-qubit Hilbert space.  It follows that
\begin{equation}
P=\ket w\!\bra w\otimes\id_2.
\end{equation}

Conversely, fix $w\in\mathcal W_Z\cup\mathcal W_X$ and set $P_w=\ket w\!\bra w\otimes\id_2$.  For every $R\in\mathcal R$,
\begin{equation}
\bra wR\ket w=0.
\label{eq:collapse-R-expectation-zero}
\end{equation}
Indeed, a computational-basis product state has zero expectation on every member of $\mathcal R$ because each contains an $X$ or $Y$ factor, while an $X$-basis product state has zero expectation because each member contains a $Y$ or $Z$ factor.  Hence, for every $Q\in\{I,X,Y,Z\}$,
\begin{equation}
P_w(R\otimes Q)P_w=0,
\end{equation}
so $P_w$ detects all of $\mathcal F^{\mathrm{bg}}$.

The monitored operators act trivially on the third qubit.  If $w\in\mathcal W_X$, then
\begin{equation}
\bra wZI\ket w
=
\bra wIZ\ket w
=
\bra wZZ\ket w
=0,
\end{equation}
so the monitored vector is the origin.  If $w\in\mathcal W_Z$, the three expectations are the joint eigenvalues of $(ZI,IZ,ZZ)$ and give exactly the four vectors $\mathbf e_{00},\mathbf e_{01},\mathbf e_{10},\mathbf e_{11}$ in Eq.~\eqref{eq:collapse-tetrahedron-vertices}.  All $36$ coordinates indexed by $\mathcal F^{\mathrm{bg}}$ vanish in both cases.  Thus $P_w$ detects every entry of $E^{\mathrm{full}}$, and the projector classification and range formulas in Eqs.~\eqref{eq:collapse-full-range}--\eqref{eq:collapse-monitored-range} follow.
\end{proof}

\subsection{Radial obstruction and disconnected signature spectrum}

Theorem~\ref{thm:five-point-range} should be compared directly with the unrestricted commuting result
\begin{equation}
\Lambda_2(E^{\mathrm{tet}})
=
T:=\conv\{\mathbf e_{00},\mathbf e_{01},\mathbf e_{10},\mathbf e_{11}\}.
\label{eq:collapse-tetrahedron}
\end{equation}
Figure~\ref{fig:tetrahedron-five-point} contrasts this unrestricted tetrahedron with the five-point configuration that survives the additional background detection constraints.  To express the contrast in sector weights, write a probability vector over the four sectors as $p=(p_{00},p_{01},p_{10},p_{11})$; its affine sector map is
\begin{equation}
\lambda(p)=\sum_{i,j}p_{ij}\mathbf e_{ij}.
\label{eq:collapse-sector-map}
\end{equation}

\begin{figure*}[t]
  \centering
  \includegraphics[width=\textwidth]{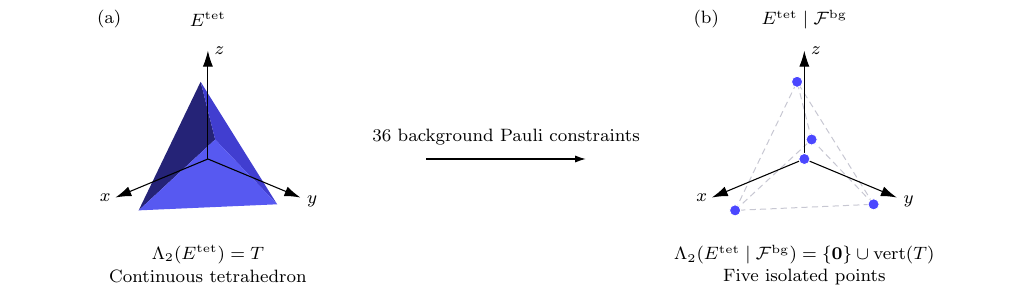}
  \caption{Effect of the background constraints on the monitored tuple $E^{\mathrm{tet}}=(Z_1,Z_2,Z_1Z_2)$. (a) Without additional background constraints, the rank-two range is the continuous tetrahedron $T$. (b) The $36$ noncommuting Pauli constraints in $\mathcal F^{\mathrm{bg}}$ reduce the monitored range to the tetrahedron barycenter and its four vertices.}
  \label{fig:tetrahedron-five-point}
\end{figure*}
The full simplex $\Delta_3$ occurs in Eq.~\eqref{eq:collapse-tetrahedron}, whereas the added constraints permit only
\begin{equation}
\left\{
\left(\frac14,\frac14,\frac14,\frac14\right),
\delta_{00},\delta_{01},\delta_{10},\delta_{11}
\right\}.
\label{eq:collapse-feasible-weights}
\end{equation}
The uniform distribution maps to the tetrahedron barycenter $0$, and the delta distributions map to its four vertices.  Thus trace-compatible sector weights alone do not determine achievability: the noncommuting background equations restrict which realizations of those weights are attainable.

This obstruction is already visible on the radial line from the barycenter to $\mathbf e_{00}$.  Let
\begin{equation}
p(t)
=(1-t)\left(\frac14,\frac14,\frac14,\frac14\right)
+t\delta_{00},
\qquad 0\le t\le1,
\label{eq:collapse-radial-weights}
\end{equation}
so $\lambda(p(t))=t(1,1,1)$.  The standard commuting realization uses
\begin{equation}
\ket{\psi_t}
=
\frac{\sqrt{1+3t}}{2}\ket{00}
+
\frac{\sqrt{1-t}}{2}
\bigl(\ket{01}+\ket{10}+\ket{11}\bigr)
\label{eq:collapse-radial-state}
\end{equation}
on the first two qubits.  One of the added constraints is $XZ\otimes X\in\mathcal F^{\mathrm{bg}}$, and its active expectation is
\begin{equation}
\bra{\psi_t}XZ\ket{\psi_t}
=
\frac12
\left[
\sqrt{(1-t)(1+3t)}-(1-t)
\right].
\label{eq:collapse-radial-obstruction}
\end{equation}
This vanishes only for $t=0$ and $t=1$.  Lemma~\ref{lem:collapse-rank-one-classification} and Theorem~\ref{thm:five-point-range} show that no alternative realization restores an interior point.

\begin{corollary}[Disconnected signature spectrum]
\label{cor:collapse-signature}
The signature spectrum of the full $39$-entry tuple, equivalently of the monitored triple under the $36$ additional constraints $\mathcal F^{\mathrm{bg}}$, is
\begin{equation}
\Sigma_2(E^{\mathrm{full}})
=
\Sigma_2(E^{\mathrm{tet}}\mid\mathcal F^{\mathrm{bg}})
=
\{0,\sqrt3\}.
\label{eq:collapse-signature}
\end{equation}
\end{corollary}

Reference~\cite{GuptaSunHeZeng2026} observed interval signature spectra in the unrestricted Pauli-detection problems studied there.  Our example shows that this behavior does not extend to all prescribed finite Pauli families.  There is no conflict with Theorem~\ref{thm:block-multiplicity}: the monitored triple has one-dimensional active blocks of multiplicity two and hence fills the tetrahedron, whereas the enlarged $39$-entry tuple does not retain this identity-on-multiplicity structure.  The $36$ additional constraints instead leave only five attainable compression vectors and open a radial gap.


\section{Discussion}
\label{sec:discussion}

We began by asking what features of an operator representation determine the geometry of its Pauli compression vectors: when can active expectation data be lifted to exact rank-$K$ scalar compressions, when is the resulting range connected or convex, and what can obstruct these properties?  The results of this paper show that representation multiplicity provides a broad answer to the first two questions, while insufficient degeneracy and additional detection constraints account for several distinct departures from that regular behavior.

Theorem~\ref{thm:block-multiplicity} makes the role of multiplicity precise. An active block state can be lifted whenever $m_\alpha\ge K\rank(\rho_\alpha)$ on every occupied block.  When this condition holds uniformly for all active states, the compression range is exactly the full state-space expectation body and is therefore compact, convex, and connected.  The two main specializations illustrate this result in familiar coding settings.  For a commuting Pauli family, multiplicity is the degeneracy of the common eigenspaces.  If that degeneracy is at least the compression rank, the range is the convex hull of the joint sector values, with an explicit projector realizing every point.  Once the compression rank exceeds the available sector degeneracy, this description can fail sharply: the range may shrink, collapse to a lower-dimensional set, or become empty.  For a noncommuting subsystem gauge algebra, the protected logical subsystem instead supplies a spectator factor.  In the single-block setting, spectator multiplicity $M\ge K$ realizes every pure-state active expectation tuple, while $M\ge KD_{\mathrm{act}}$ realizes the full active state-space body.  This active/spectator picture accounts for the spheres, balls, disks, intervals, and product bodies found in the subsystem examples.  The parent-Hamiltonian construction gives the same structure a dynamical interpretation by selecting particular points or paths within these bodies.

The signature spectrum poses a related but genuinely different question. Proposition~\ref{prop:semialgebraic-components} shows that, for any finite background-constrained compression problem, the signature spectrum is a finite union of closed intervals.  Thus a connected compression range always has an interval as its radial image, but a disconnected vector range may do so as well when the radial images of its components overlap.  The five-point construction shows that such overlap is not automatic.  There, the additional noncommuting detection equations reduce an otherwise continuous tetrahedral range to its barycenter and four vertices, producing the disconnected spectrum $\{0,\sqrt{3}\}$. Consequently, interval behavior is not universal for arbitrary prescribed finite Pauli families.  At the same time, the example is deliberately structured through a special collection of background constraints, so it does not settle the corresponding question for the more systematic error families that usually arise in quantum coding theory.

Several natural problems remain.  The most immediate one is to determine whether signature spectra are always intervals for Pauli families defined by a weight cutoff, locality condition, or another coding-theoretically natural rule.  On the vector side, sharper criteria are needed below the multiplicity and common-sector degeneracy thresholds, particularly in the presence of background detection equations: when must the constrained range remain connected, star-shaped, or convex, and when can it become finite or disconnected?  On the radial side, it would be useful to identify structural conditions forcing the component intervals in Eq.~\eqref{eq:finite-union-radial-intervals} to overlap.  Finally, reducing the number, weight, or Clifford complexity of the $39$-entry counterexample may reveal the smallest Pauli configuration capable of producing a genuine radial gap.

\medskip
\noindent\textbf{Note added.} The three-qubit example with a disconnected signature spectrum in Sec.~\ref{sec:collapse} was initially identified by GPT-5.6 and subsequently incorporated by the authors into the mathematical framework and discussion developed here. AI tools were also used to assist with writing and revising the manuscript.


\appendix

\section{Rank-sensitive diagonal ranges beyond sector multiplicity}
\label{app:rank-sensitive}

The general diagonal lower and upper bounds below prove the rank-sensitive range classification for $E^Z=(Z_1,Z_2,Z_3)$ stated in Proposition~\ref{prop:three-Z-ranges} of Sec.~\ref{subsec:rank-sensitive-summary}.

\begin{proposition}[Diagonal disjoint-support lower bound and half-space upper bound]
\label{prop:diagonal-bounds}
Let \(A=(A_1,\ldots,A_m)\) be a tuple of real diagonal \(N\times N\) matrices,
\begin{align}
A_i&=\operatorname{diag}(a_i(1),\ldots,a_i(N)),\nonumber\\
\mathbf b_q&:=\bigl(a_1(q),\ldots,a_m(q)\bigr)\in\R^m .
\label{eq:diag-eigenvalue-tuples}
\end{align}
For \(1\le K\le N\), the following hold.

For the lower bound, let \(J_1,\ldots,J_K\) be nonempty pairwise disjoint
subsets of \(\{1,\ldots,N\}\).  Then
\begin{equation}
\bigcap_{a=1}^{K}\conv\{\mathbf b_q:\ q\in J_a\}
\subseteq
\Lambda_K(A).
\label{eq:diag-lower-bound}
\end{equation}

For the upper bound, let \(c\in\R\) and \(h\in\R^m\), and write
\begin{equation}
H(h,c):=\{\lambda\in\R^m:\ h\cdot\lambda\le c\}.
\end{equation}
Let \(\mathfrak H_K(A)\) be the collection of all such half-spaces that contain at least \(N-K+1\) of the indexed points \(\mathbf b_q\).  Then
\begin{equation}
\Lambda_K(A)
\subseteq
\bigcap_{H\in\mathfrak H_K(A)} H .
\label{eq:diag-upper-bound}
\end{equation}
\end{proposition}

\begin{proof}
For the lower bound, let \(\lambda=(\lambda_1,\ldots,\lambda_m)\) lie in the intersection in
Eq.~\eqref{eq:diag-lower-bound}.  For each \(a\), choose coefficients
\(t_{a,q}\ge0\), supported on \(J_a\), such that
\begin{equation}
\sum_{q\in J_a}t_{a,q}=1,
\qquad
\sum_{q\in J_a}t_{a,q}\mathbf b_q=\lambda .
\end{equation}
With respect to the standard basis \(\{e_q\}_{q=1}^N\) of \(\C^N\), define
\begin{equation}
u_a:=\sum_{q\in J_a}\sqrt{t_{a,q}}\,e_q .
\end{equation}
The supports \(J_a\) are disjoint, so \(u_1,\ldots,u_K\) are orthonormal.  Since the \(A_i\) are diagonal,
\begin{equation}
\langle u_a,A_i u_{a'}\rangle
=\delta_{aa'}\sum_{q\in J_a}t_{a,q}a_i(q)
=\delta_{aa'}\lambda_i .
\end{equation}
Thus the rank-\(K\) projector \(P\) onto \(\operatorname{span}\{u_1,\ldots,u_K\}\) satisfies
\begin{equation}
P A_i P=\lambda_iP,\qquad i=1,\ldots,m,
\end{equation}
which proves the lower inclusion.

For the upper bound, suppose \(P\in\cP_K(\C^N)\) satisfies \(PA_iP=\lambda_iP\).  Let \(H(h,c)\in\mathfrak H_K(A)\), and choose an index set \(Q\subseteq\{1,\ldots,N\}\) with \(|Q|=N-K+1\) such that \(\mathbf b_q\in H(h,c)\) for \(q\in Q\).  The coordinate subspace \(S_Q:=\operatorname{span}\{e_q:\ q\in Q\}\) has dimension \(N-K+1\), hence
\begin{equation}
\dim\bigl(\operatorname{Ran}(P)\cap S_Q\bigr)
\ge K+(N-K+1)-N=1 .
\end{equation}
Choose a unit vector \(u\) in this intersection.  Then
\begin{equation}
\lambda_i
=\langle u,A_i u\rangle
=\sum_{q\in Q}|\langle e_q,u\rangle|^2a_i(q),
\qquad i=1,\ldots,m.
\end{equation}
The numbers \(|\langle e_q,u\rangle|^2\), \(q\in Q\), are nonnegative and sum to one, so
\begin{equation}
\lambda=\sum_{q\in Q}|\langle e_q,u\rangle|^2\mathbf b_q
\in\conv\{\mathbf b_q:\ q\in Q\}.
\end{equation}
Since a half-space is convex and contains all indexed points \(\mathbf b_q\), \(q\in Q\), this convex hull is contained in \(H(h,c)\).  Hence \(\lambda\in H(h,c)\). As this holds for every \(H\in\mathfrak H_K(A)\), \(\lambda\) lies in the intersection in Eq.~\eqref{eq:diag-upper-bound}.
\end{proof}

We now apply these bounds to prove Proposition~\ref{prop:three-Z-ranges}.  For brevity write \(E:=E^Z=(Z_1,Z_2,Z_3)\).  Its joint diagonal eigenvalue points are the eight vertices
\begin{equation}
\mathcal V:=\{(\epsilon_1,\epsilon_2,\epsilon_3):\ \epsilon_i\in\{\pm1\}\}
\end{equation}
of the cube.  Since every common eigenspace is one-dimensional, the sector multiplicity is \(M=1\).  The rank-one identity in Eq.~\eqref{eq:three-Z-ranges} follows directly from Theorem~\ref{thm:polytope}; it remains to prove the higher-rank claims.

For \(\Lambda_2\), split the cube vertices into the two parity tetrahedra
\begin{align}
\mathcal V_+
&:=\{\epsilon\in\mathcal V:\epsilon_1\epsilon_2\epsilon_3=1\},\nonumber\\
\mathcal V_-
&:=\{\epsilon\in\mathcal V:\epsilon_1\epsilon_2\epsilon_3=-1\}.
\label{eq:cube-parity-tetrahedra}
\end{align}
If \(|x|+|y|+|z|\le1\), the four numbers
\begin{equation}
\begin{gathered}
\frac{1+x+y+z}{4},\quad
\frac{1+x-y-z}{4},\\
\frac{1-x+y-z}{4},\quad
\frac{1-x-y+z}{4}
\end{gathered}
\end{equation}
are nonnegative, because each numerator is bounded below by \(1-|x|-|y|-|z|\).  They sum to one and are the barycentric coefficients of \((x,y,z)\) for the ordered vertices
\begin{equation}
(1,1,1),\quad (1,-1,-1),\quad (-1,1,-1),\quad (-1,-1,1)
\end{equation}
of \(\mathcal V_+\).  The analogous four coefficients
\begin{equation}
\begin{gathered}
\frac{1-x-y-z}{4},\quad
\frac{1-x+y+z}{4},\\
\frac{1+x-y+z}{4},\quad
\frac{1+x+y-z}{4}
\end{gathered}
\end{equation}
are again nonnegative, sum to one, and are the barycentric coefficients for the ordered vertices
\begin{equation}
(-1,-1,-1),\quad (-1,1,1),\quad (1,-1,1),\quad (1,1,-1)
\end{equation}
of \(\mathcal V_-\).  Thus the octahedron is contained in both parity tetrahedra.  The lower bound in Proposition~\ref{prop:diagonal-bounds}, with \(J_1,J_2\) the disjoint computational-basis index sets carrying eigenvalue sets \(\mathcal V_+\) and \(\mathcal V_-\), therefore realizes the entire octahedron in \(\Lambda_2(E)\).

Conversely, for every cube vertex \(\epsilon\in\mathcal V\), the half-space
\begin{equation}
H_\epsilon:=\{\lambda:\epsilon\cdot\lambda\le1\}
\end{equation}
contains exactly the seven cube vertices different from \(\epsilon\).  Since \(N-K+1=7\) for \(N=8\) and \(K=2\), each \(H_\epsilon\) belongs to \(\mathfrak H_2(E)\).  The upper bound gives
\begin{equation}
\Lambda_2(E)
\subseteq
\bigcap_{\epsilon\in\mathcal V}\{\lambda:\epsilon\cdot\lambda\le1\}
=\{(x,y,z): |x|+|y|+|z|\le1\}.
\end{equation}
The last equality follows from \(\max_{\epsilon\in\mathcal V}\epsilon\cdot(x,y,z)=|x|+|y|+|z|\).  Combining the lower and upper inclusions proves the rank-two identity in Eq.~\eqref{eq:three-Z-ranges}.

For ranks \(3\) and \(4\), the upper bound with the half-spaces
\begin{equation}
H_{pq}^{\eta}:=\{\lambda:\ \eta_p\lambda_p+\eta_q\lambda_q\le0\}
\end{equation}
where \(1\le p<q\le3\) and \(\eta_p,\eta_q\in\{\pm1\}\), gives the inclusion
\begin{equation}
\Lambda_K(E)\subseteq\bigcap_{p,q,\eta} H_{pq}^{\eta} 
\end{equation}
for \(K=3,4\).  At the cube vertices, the value of \(\eta_p\epsilon_p+\eta_q\epsilon_q\) is \(2,0\), or \(-2\), and the inequality fails only when \(\epsilon_p=\eta_p\) and \(\epsilon_q=\eta_q\).  Hence each \(H_{pq}^{\eta}\) contains six vertices.  Since here \(N=8\), the required threshold \(N-K+1\) equals \(6\) for \(K=3\) and \(5\) for \(K=4\), so each \(H_{pq}^{\eta}\) belongs to \(\mathfrak H_K(E)\) in both cases.

For a fixed pair \((p,q)\), intersecting the four sign choices gives
\begin{equation}
|\lambda_p|+|\lambda_q|
=
\max_{\eta_p,\eta_q\in\{\pm1\}}(\eta_p\lambda_p+\eta_q\lambda_q)
\le0.
\end{equation}
Writing \(\lambda=(x,y,z)\), the two pairs \((1,2)\) and \((1,3)\) give explicitly
\begin{equation}
|x|+|y|\le0,\qquad |x|+|z|\le0.
\end{equation}
Thus \(x=y=z=0\), and \(\Lambda_3(E)\) and \(\Lambda_4(E)\) are both contained in \(\{0\}\). 

For the reverse inclusion, partition the cube vertices into the four antipodal pairs $\{\epsilon,-\epsilon\}$. The convex hull of each pair contains the origin.  Choosing these four pairs as disjoint index sets in the lower bound of Proposition~\ref{prop:diagonal-bounds} gives \(0\in\Lambda_4(E)\).  Since any rank-\((K+1)\) scalar compression restricts to a rank-\(K\) scalar compression on a subspace, \(\Lambda_{K+1}(E)\subseteq\Lambda_K(E)\); hence \(0\in\Lambda_3(E)\) as well.  Combining this with the upper inclusion gives
\begin{equation}
\Lambda_3(E)=\Lambda_4(E)=\{(0,0,0)\}.
\end{equation}
This proves the rank-three and rank-four identities in Eq.~\eqref{eq:three-Z-ranges}.

Finally, if a rank-\(K\) joint scalar compression existed for \(E\) with \(K\ge5\), then its first coordinate would give a rank-\(K\) scalar compression of the single operator \(Z_1\).  Write
\begin{equation}
\C^8=V_+\oplus V_-,
\qquad
Z_1|_{V_+}=+\id,\qquad
Z_1|_{V_-}=-\id .
\end{equation}
Here \(\dim V_+=\dim V_-=4\).  If
\begin{equation}
PZ_1P=\lambda P,
\qquad
W=\operatorname{Ran}P,
\end{equation}
then the Hermitian form vanishes identically on \(W\):
\begin{equation}
\langle u,(Z_1-\lambda \id)v\rangle=0,
\qquad u,v\in W .
\end{equation}
Taking \(u=v\) with \(\|u\|=1\) first gives
\begin{equation}
\lambda=\langle u,Z_1u\rangle\in[-1,1].
\end{equation}
For \(-1<\lambda<1\), write
\begin{equation}
u=u_++u_-,
\qquad
u_\pm\in V_\pm .
\end{equation}
The identity
\begin{equation}
0=\langle u,(Z_1-\lambda\id)u\rangle
=(1-\lambda)\|u_+\|^2-(1+\lambda)\|u_-\|^2
\end{equation}
shows in particular that no nonzero vector of \(W\) can have \(u_+=0\).  Hence the projection \(W\to V_+\) is injective, and
\begin{equation}
\dim W\le\dim V_+=4 .
\end{equation}
For \(\lambda=1\), the same identity gives \(u_-=0\), so \(W\subseteq V_+\); for \(\lambda=-1\), it gives \(u_+=0\), so \(W\subseteq V_-\).  Again \(\dim W\le4\).  Thus no such \(W\) can have dimension \(5\) or larger.  Therefore
\begin{equation}
\Lambda_K(E)=\emptyset\qquad(K\ge5).
\end{equation}


\section{Calculation of \texorpdfstring{$\Lambda_2(X\otimes \id_M,\;Y\otimes \id_M,\;Z\otimes \id_M)$}{Lambda2(X tensor IM, Y tensor IM, Z tensor IM)}}
\label{app:rank2-bloch}

To support the sphere-to-ball transition in Proposition~\ref{prop:412-atlas} of Sec.~\ref{subsec:atlas412}, let $X,Y,Z$ denote the Pauli matrices on $\C^2$ and $\id_M$ the $M\times M$ identity.  We compute the rank-$2$ joint numerical range of the triple
\begin{equation}
E^{(M)}:=(X\otimes \id_M,\ Y\otimes \id_M,\ Z\otimes \id_M).
\end{equation}

\begin{proposition}[Rank-$2$ Pauli sphere/ball transition]
\label{prop:rank2-pauli-ball}
For \(M=2,3\),
\begin{equation}
\begin{aligned}
\Lambda_2(E^{(M)})
&=\Lambda_1(X,Y,Z)\\
&=\bigl\{(x,y,z)\in\R^3:\ x^2+y^2+z^2=1\bigr\}.
\end{aligned}
\end{equation}
For \(M\ge4\),
\begin{equation}
\begin{aligned}
\Lambda_2(E^{(M)})
&=\conv\Lambda_1(X,Y,Z)\\
&=\bigl\{(x,y,z)\in\R^3:\ x^2+y^2+z^2\le 1\bigr\}.
\end{aligned}
\end{equation}
\end{proposition}

\begin{proof}
For any $M\ge 2$ one has the elementary inclusions
\begin{equation}
\Lambda_1(X,Y,Z)\subseteq\Lambda_2(E^{(M)})\subseteq\conv\Lambda_1(X,Y,Z).
\label{eq:basic-inclusions}
\end{equation}
The left inclusion holds because any pure qubit state $\ket{u}\in\C^2$ can be lifted to an orthonormal pair in \(\C^2\otimes\C^M\), for instance \(\ket{u}\otimes\ket{e_1}\) and \(\ket{u}\otimes\ket{e_2}\), producing the same expectation values. For the right inclusion, take the normalized state \(\rho=P/2\) associated with any rank-$2$ scalar compression \(P\).  The compression value is the expectation tuple of the reduced qubit state \(\Tr_{\C^M}\rho\), hence belongs to the convex hull of the pure-state Bloch sphere.

For \(M=3\), let
\(\lambda=(\lambda_1,\lambda_2,\lambda_3)\in\Lambda_2(E^{(3)})\).
Then there exist orthonormal $\ket{x},\ket{y}\in\C^2\otimes\C^3$ such that
\begin{equation}
\begin{aligned}
\bra{x}(X\otimes \id_3)\ket{x}=\bra{y}(X\otimes \id_3)\ket{y}&=\lambda_1,\\
\bra{x}(Y\otimes \id_3)\ket{x}=\bra{y}(Y\otimes \id_3)\ket{y}&=\lambda_2,\\
\bra{x}(Z\otimes \id_3)\ket{x}=\bra{y}(Z\otimes \id_3)\ket{y}&=\lambda_3,
\end{aligned}
\end{equation}
together with the orthogonality constraints
\begin{equation}
\begin{aligned}
\braket{y}{x}&=0,
&\qquad \bra{y}(X\otimes \id_3)\ket{x}&=0,\\
\bra{y}(Y\otimes \id_3)\ket{x}&=0,
&\qquad \bra{y}(Z\otimes \id_3)\ket{x}&=0.
\end{aligned}
\end{equation}
Write
\begin{equation}
\ket{x}=\ket{0}\otimes x_1+\ket{1}\otimes x_2,
\qquad
\ket{y}=\ket{0}\otimes y_1+\ket{1}\otimes y_2,
\end{equation}
with \(x_1,x_2,y_1,y_2\in\C^3\).  The four off-diagonal equations are equivalent to
\begin{equation}
\langle y_i,x_j\rangle=0\qquad (i,j\in\{1,2\}),
\end{equation}
because \(I,X,Y,Z\) form a basis for the \(2\times2\) matrices on the qubit factor.  Thus \(\operatorname{span}\{x_1,x_2\}\) is orthogonal to \(\operatorname{span}\{y_1,y_2\}\).  Since both subspaces lie in \(\C^3\), at least one of the two pairs is linearly dependent.

If \(\{x_1,x_2\}\) is linearly dependent, then \(\ket{x}= \ket{u}\otimes\ket{v}\) for some unit vectors \(\ket{u}\in\C^2\) and \(\ket{v}\in\C^3\), and hence
\begin{equation}
\lambda
=\bigl(\bra{u}X\ket{u},\bra{u}Y\ket{u},\bra{u}Z\ket{u}\bigr)
\in\Lambda_1(X,Y,Z).
\end{equation}
The case in which \(\{y_1,y_2\}\) is dependent is identical, using \(\ket{y}\) instead of \(\ket{x}\). Combining this with Eq.~\eqref{eq:basic-inclusions} yields
\begin{equation}
\Lambda_2(E^{(3)})=\Lambda_1(X,Y,Z),
\end{equation}
which proves the sphere case for \(M=3\).  The case \(M=2\) follows by embedding \(\C^2\otimes\C^2\) into \(\C^2\otimes\C^3\): together with Eq.~\eqref{eq:basic-inclusions}, this gives
\begin{equation}
\Lambda_2(E^{(2)})=\Lambda_1(X,Y,Z).
\end{equation}

For \(M\ge4\), every qubit density matrix can be purified
in \(\C^2\otimes\C^2\), so
\begin{equation}
\conv\Lambda_1(X,Y,Z)=\Lambda_1(X\otimes \id_2,\ Y\otimes \id_2,\ Z\otimes \id_2).
\end{equation}
Lifting a pure state in \(\C^2\otimes\C^2\) to a rank-$2$ product subspace in \((\C^2\otimes\C^2)\otimes\C^2\cong\C^2\otimes\C^4\) gives
\begin{equation}
\begin{aligned}
\conv\Lambda_1(X,Y,Z)
&=\Lambda_1(X\otimes \id_2,\ Y\otimes \id_2,\ Z\otimes \id_2)\\
&\subseteq\Lambda_2(X\otimes \id_4,\ Y\otimes \id_4,\ Z\otimes \id_4)\\
&\subseteq\Lambda_2(E^{(M)}),
\end{aligned}
\end{equation}
where the last inclusion uses an isometric embedding \(\C^4\hookrightarrow\C^M\).  Together with Eq.~\eqref{eq:basic-inclusions}, this gives
\begin{equation}
\Lambda_2(X\otimes \id_M,\ Y\otimes \id_M,\ Z\otimes \id_M)
=\conv\Lambda_1(X,Y,Z)
\end{equation}
for all $M\ge 4$.  The formula
\begin{equation}
\Lambda_1(X,Y,Z)
=
\{(x,y,z)\in\R^3:x^2+y^2+z^2=1\}
\end{equation}
describes the Bloch sphere, whose convex hull is the unit ball.  Therefore, the rank-$2$ joint numerical range of $(X\otimes \id_M,Y\otimes \id_M,Z\otimes \id_M)$ is the sphere for $M=2,3$ and the ball for $M\ge 4$.
\end{proof}

\begin{proposition}[Rank-$2$ Pauli disks and intervals]
\label{prop:rank2-pauli-lower-charts}
For every $M\ge2$ and every pair of distinct Pauli axes $A,B\in\{X,Y,Z\}$,
\begin{equation}
\Lambda_2(A\otimes\id_M,B\otimes\id_M)
=
\{(a,b)\in\R^2:a^2+b^2\le1\}.
\end{equation}
For every $A\in\{X,Y,Z\}$,
\begin{equation}
\Lambda_2(A\otimes\id_M)=[-1,1].
\end{equation}
\end{proposition}

\begin{proof}
Dropping one coordinate from a scalar compression of the Pauli triple gives a scalar compression of the corresponding pair.  By Proposition~\ref{prop:rank2-pauli-ball}, the coordinate projection of the triple range is the full unit disk for every $M\ge2$, whether the triple range is the sphere or the ball.  Hence the unit disk is contained in the two-axis range.

Conversely, let $P$ be a rank-$2$ scalar compression for a pair $(A\otimes\id_M,B\otimes\id_M)$, set $\rho=P/2$, and let $\sigma=\Tr_{\C^M}\rho$ be its reduced qubit state.  The compression values satisfy
\begin{equation}
a=\Tr(\sigma A),
\qquad
b=\Tr(\sigma B).
\end{equation}
The Bloch-vector bound for $\sigma$ gives $a^2+b^2\le1$, proving the pair formula.  Projecting the disk onto either coordinate gives the full interval $[-1,1]$, while the operator-norm bound gives the reverse inclusion for a single axis.
\end{proof}


\section{A product-body computation: \texorpdfstring{$\Lambda_4(Z_1,\;X_3,\;Z_3)$}{Lambda4(Z1, X3, Z3)}}
\label{app:cylinder}

The calculation below establishes the interval$\times$disk identity in Eq.~\eqref{eq:stabilizer-gauge-cylinder} of Sec.~\ref{subsec:atlas412}, obtained by retaining one commuting stabilizer coordinate together with two noncommuting Pauli coordinates on one gauge qubit.

\begin{proposition}[Interval$\times$disk cylinder]
\label{prop:cylinder}
On four qubits one has
\begin{equation}
\Lambda_4(Z_1,X_3,Z_3)
=\bigl\{(x,y,z)\in\R^3:\ |x|\le 1, y^2+z^2\le 1\bigr\}.
\end{equation}
\end{proposition}

\begin{proof}
First note that swapping tensor factors $2$ and $3$ is implemented by a unitary $W$ with
\begin{equation}
W^\dagger Z_1 W=Z_1,
\qquad
W^\dagger X_3 W=X_2,
\qquad
W^\dagger Z_3 W=Z_2.
\end{equation}
By unitary invariance of joint higher-rank numerical ranges,
\begin{equation}
\Lambda_4(Z_1,X_3,Z_3)=\Lambda_4(Z_1,X_2,Z_2).
\label{eq:cyl-swap}
\end{equation}
Grouping qubits as $(1,2)$ versus $(3,4)$ gives
\begin{equation}
\begin{aligned}
Z_1&=Z\otimes \id_2\otimes \id_4,\\
X_2&=\id_2\otimes X\otimes \id_4,\\
Z_2&=\id_2\otimes Z\otimes \id_4,
\end{aligned}
\end{equation}
where $\id_4$ acts on the spectator pair of qubits $(3,4)$.

Define the two-qubit triple
\begin{equation}
(B_1,B_2,B_3)
:=
\bigl(Z\otimes \id_2,\ \id_2\otimes X,\ \id_2\otimes Z\bigr)
\quad \text{on } \C^2\otimes \C^2.
\end{equation}
We first compute its rank-$1$ joint numerical range:
\begin{equation}
\Lambda_1(B_1,B_2,B_3)
=
\bigl\{(x,y,z)\in\R^3:\ |x|\le 1,\ \ y^2+z^2\le 1\bigr\}.
\label{eq:cyl-lambda1}
\end{equation}
Indeed, for any unit vector $\ket{\psi}\in\C^2\otimes\C^2$, the first coordinate $x=\bra{\psi}(Z\otimes \id_2)\ket{\psi}$ lies in $[-1,1]$.  The remaining pair satisfies
\begin{equation}
(y,z)\in \Lambda_1(\id_2\otimes X,\ \id_2\otimes Z)=\Lambda_1(X,Z),
\end{equation}
and
\begin{equation}
\Lambda_1(X,Z)=\{(y,z)\in\R^2:\ y^2+z^2\le 1\}.
\end{equation}
This proves the inclusion ``$\subseteq$'' in Eq.~\eqref{eq:cyl-lambda1}.

Conversely, fix any $(x,y,z)$ with $|x|\le 1$ and $y^2+z^2\le 1$.  Choose a single-qubit pure state $\ket{\phi}\in\C^2$ such that
\begin{equation}
(\bra{\phi}X\ket{\phi},\bra{\phi}Z\ket{\phi})=(y,z).
\end{equation}
Such a state exists because $\Lambda_1(X,Z)$ is exactly the unit disk.  Let $a=(1+x)/2$ and $b=(1-x)/2$ and define a product two-qubit state
\begin{equation}
\ket{\psi}=\bigl(\sqrt{a}\,\ket{0}+\sqrt{b}\,\ket{1}\bigr)\otimes \ket{\phi}.
\end{equation}
A direct calculation shows
\begin{equation}
\begin{aligned}
\bra{\psi}(Z\otimes \id_2)\ket{\psi}&=a-b=x,\\
\bra{\psi}(\id_2\otimes X)\ket{\psi}&=\bra{\phi}X\ket{\phi}=y,\\
\bra{\psi}(\id_2\otimes Z)\ket{\psi}&=\bra{\phi}Z\ket{\phi}=z.
\end{aligned}
\end{equation}
Hence $(x,y,z)\in\Lambda_1(B_1,B_2,B_3)$, which proves the inclusion ``$\supseteq$'' in Eq.~\eqref{eq:cyl-lambda1}.

Finally, relate Eq.~\eqref{eq:cyl-lambda1} to the rank-$4$ numerical range on four qubits.  For any $\ket{\psi}\in\C^2\otimes\C^2$, the rank-$4$ projector
\begin{equation}
P=\ket{\psi}\!\bra{\psi}\otimes \id_4
\end{equation}
satisfies
\begin{equation}
P(B_j\otimes \id_4)P=\lambda_j P,\qquad
\lambda_j=\bra{\psi}B_j\ket{\psi}.
\end{equation}
Therefore,
\begin{equation}
\Lambda_1(B_1,B_2,B_3)
\subseteq
\Lambda_4(B_1\otimes \id_4,\ B_2\otimes \id_4,\ B_3\otimes \id_4).
\end{equation}
Conversely, if \(P\) is any rank-$4$ projector with \(P(B_j\otimes \id_4)P=\lambda_jP\), then the normalized state \(\rho=P/4\) yields
\begin{equation}
\lambda_j=\Tr(\rho(B_j\otimes \id_4))
=\Tr(\rho_{12}B_j),
\end{equation}
where $\rho_{12}:=\Tr_{34}\rho$ is a two-qubit density matrix. Thus $(\lambda_1,\lambda_2,\lambda_3)$ lies in the set of expectation tuples over mixed two-qubit states, which is the convex hull of $\Lambda_1(B_1,B_2,B_3)$. Since the cylinder in Eq.~\eqref{eq:cyl-lambda1} is convex, this convex hull equals $\Lambda_1(B_1,B_2,B_3)$ itself. Combining these inclusions, and using Eq.~\eqref{eq:cyl-swap}, establishes the claim.
\end{proof}


\section{Additional parent-Hamiltonian realizations}
\label{app:parent-hamiltonian-examples}

The ground-space mechanism of Sec.~\ref{sec:hamiltonians} admits three further realizations: a tensor-product extension, a cluster realization, and a Clifford-covariant mixed-Pauli construction.  Each retains the same separation between active dynamics and spectator multiplicity.

\subsection{Two active qubits from disjoint pairs, \texorpdfstring{$n=4$}{n=4}}
\label{ex:n4-r2}
Consider
\begin{equation}
H(\theta)=-(1-\theta)(Z_1Z_2+Z_3Z_4)-\theta(X_1+X_4),
\end{equation}
with spectator symmetries generated by \(Z_2\) and \(Z_3\).  In the sector \((Z_2,Z_3)=(\zeta_2,\zeta_3)\in\{\pm1\}^2\), the active Hamiltonian on qubits \(1\) and \(4\) is
\begin{equation}
H_{\zeta_2,\zeta_3}(\theta)
=-(1-\theta)(\zeta_2 Z_1+\zeta_3 Z_4)-\theta(X_1+X_4).
\end{equation}
This is a sum of two independent copies of the minimal one-qubit active interpolation described in Sec.~\ref{sec:hamiltonians}.  Each sector has the same unique active ground energy \(-2\varepsilon(\theta)\), and the four sign choices are related by conjugation with \(X_1\) and/or \(X_4\).  Hence the full ground projector has rank \(4\).  Moreover,
\begin{equation}
\begin{aligned}
P_\theta(Z_1Z_2)P_\theta&=\frac{1-\theta}{\varepsilon(\theta)}\,P_\theta,&
P_\theta(Z_3Z_4)P_\theta&=\frac{1-\theta}{\varepsilon(\theta)}\,P_\theta,\\
P_\theta X_1P_\theta&=\frac{\theta}{\varepsilon(\theta)}\,P_\theta,&
P_\theta X_4P_\theta&=\frac{\theta}{\varepsilon(\theta)}\,P_\theta.
\end{aligned}
\end{equation}
Thus the ground path lies in
\begin{equation}
\Lambda_4(Z_1Z_2,\ Z_3Z_4,\ X_1,\ X_4)
\end{equation}
as the diagonal curve
\begin{equation}
\left(
\frac{1-\theta}{\varepsilon(\theta)},\frac{1-\theta}{\varepsilon(\theta)},
\frac{\theta}{\varepsilon(\theta)},\frac{\theta}{\varepsilon(\theta)}
\right).
\end{equation}
Pinning either \(Z_2\) or \(Z_3\) halves the ground-space rank; pinning both leaves a single active product ground state.

\subsection{Cluster mid-bond as a two-gauge-qubit model, \texorpdfstring{$n=4$}{n=4}}
\label{ex:cluster}
Define $K_2=Z_1X_2Z_3$ and $K_3=Z_2X_3Z_4$, and consider
\begin{equation}
H_{\mathrm{cl}}(\theta)=-(1-\theta)(K_2+K_3)-\theta(X_2+X_3).
\end{equation}
The endpoint operators \(Z_1\) and \(Z_4\) commute with every term, and hence label spectator sectors.  In a fixed sector \((Z_1,Z_4)=(\zeta_1,\zeta_4)\), the active Hamiltonian on qubits \(2,3\) is
\begin{equation}
H_{\mathrm{cl}}(\theta)\big|_{(\zeta_1,\zeta_4)}
=-(1-\theta)(\zeta_1 X_2Z_3+\zeta_4 Z_2X_3)-\theta(X_2+X_3).
\end{equation}
The four active Pauli terms have the same symplectic commutation graph as \((Z_2,Z_3,X_3,X_2)\).  Indeed, among
\begin{equation}
\zeta_1 X_2Z_3,\quad \zeta_4 Z_2X_3,\quad X_2,\quad X_3,
\end{equation}
the only anticommuting pairs are \((\zeta_1 X_2Z_3,X_3)\) and \((\zeta_4 Z_2X_3,X_2)\), matching the anticommuting pairs \((Z_2,X_2)\) and \((Z_3,X_3)\) after the ordered identification below. Hence a two-qubit Clifford on the active block, followed if necessary by sector-dependent Pauli sign corrections, sends
\begin{equation}
(\zeta_1 X_2Z_3,\ \zeta_4 Z_2X_3,\ X_2,\ X_3)\mapsto (Z_2,\ Z_3,\ X_3,\ X_2).
\end{equation}
Thus each spectator sector is Clifford-equivalent to
\begin{equation}
\widetilde H_{\mathrm{cl}}(\theta)
=-(1-\theta)(Z_2+Z_3)-\theta(X_2+X_3),
\end{equation}
up to the harmless interchange of the two transverse-field labels.  This is again a pair of independent one-qubit interpolations, so every spectator sector has the same unique active ground energy \(-2\varepsilon(\theta)\).  The full ground projector has rank \(4\), and the original monitored tuple satisfies
\begin{equation}
\begin{aligned}
P_\theta K_2P_\theta&=\frac{1-\theta}{\varepsilon(\theta)}P_\theta,&
P_\theta K_3P_\theta&=\frac{1-\theta}{\varepsilon(\theta)}P_\theta,\\
P_\theta X_2P_\theta&=\frac{\theta}{\varepsilon(\theta)}P_\theta,&
P_\theta X_3P_\theta&=\frac{\theta}{\varepsilon(\theta)}P_\theta.
\end{aligned}
\end{equation}
Thus the cluster mid-bond interpolation realizes the same diagonal quarter-circle as the disjoint-pair model, but with the two active qubits related by a nontrivial Clifford gauge choice in each spectator sector.

\subsection{Mixed-Pauli commuting families on \texorpdfstring{$n=5$}{n=5}}
\label{ex:5q-mixed}
The examples above use \(Z\)-type spectator symmetries in the displayed coordinates, but the same mechanism does not depend on the Pauli types.  Let \(G=\langle g_1,\ldots,g_4\rangle\) and \(F=\langle f_1,\ldots,f_4\rangle\) be two abelian Pauli families on five qubits, and define their phase-free label spaces by
\begin{align}
\overline G&:=\operatorname{span}_{\mathbb F_2}
\{\pi(g_1),\ldots,\pi(g_4)\},\nonumber\\
\overline F&:=\operatorname{span}_{\mathbb F_2}
\{\pi(f_1),\ldots,\pi(f_4)\}.
\end{align}
Assume that \(\overline G\) and \(\overline F\) are four-dimensional and that the pairing
\(\omega|_{\overline G\times\overline F}\) is nondegenerate.  Since the two families are abelian, both label spaces are isotropic.  We may therefore choose dual bases
\(\bar g'_1,\ldots,\bar g'_4\) of \(\overline G\) and
\(\bar f'_1,\ldots,\bar f'_4\) of \(\overline F\) such that
\begin{equation}
\omega(\bar g'_i,\bar f'_j)=\delta_{ij},\qquad
\omega(\bar g'_i,\bar g'_j)=\omega(\bar f'_i,\bar f'_j)=0.
\label{eq:five-qubit-dual-label-bases}
\end{equation}
Choose Hermitian Pauli lifts \(g'_i\) and \(f'_i\) of these labels, and set
\begin{equation}
W:=\operatorname{span}_{\mathbb F_2}
\{\bar g'_1,\bar f'_1,\ldots,\bar g'_4,\bar f'_4\}.
\end{equation}
The Gram matrix in Eq.~\eqref{eq:five-qubit-dual-label-bases} is \(J_4\), so \(W\) is an eight-dimensional nondegenerate symplectic subspace of \(\overline{\mathcal P}_5\).  Consequently,
\begin{equation}
\overline{\mathcal P}_5=W\oplus W^{\perp_\omega},
\qquad
\dim_{\mathbb F_2}W^{\perp_\omega}=2,
\end{equation}
and \(W^{\perp_\omega}\) is a nondegenerate symplectic subspace representing one spectator qubit.  Choosing a symplectic pair in \(W^{\perp_\omega}\) completes the displayed dual bases to a symplectic basis of the full five-qubit label space.  Hence, by the symplectic normal-form argument used in Lemma~\ref{lem:canonical}, there is a Clifford \(U\), after harmless Pauli sign corrections, such that
\begin{equation}
\begin{aligned}
Ug'_iU^\dagger&=Z_i\otimes\id_{\mathrm{spect}},\\
Uf'_iU^\dagger&=X_i\otimes\id_{\mathrm{spect}},
\qquad i=1,\ldots,4.
\end{aligned}
\end{equation}
For the interpolating Hamiltonian
\begin{equation}
H(\theta)=-(1-\theta)\sum_{i=1}^4 g'_i-\theta\sum_{i=1}^4 f'_i,
\end{equation}
one obtains
\begin{equation}
U H(\theta) U^\dagger=\Bigl[-(1-\theta)\sum_{i=1}^4 Z_i-\theta\sum_{i=1}^4 X_i\Bigr]\otimes\id_{\mathrm{spect}},
\end{equation}
so the active ground state is a product of four identical one-qubit ground states and is unique for every \(\theta\).  Hence the ground projector has rank \(2\).  With \(\varepsilon(\theta)=\sqrt{(1-\theta)^2+\theta^2}\),
\begin{equation}
\begin{aligned}
P_\theta g'_iP_\theta&=\frac{1-\theta}{\varepsilon(\theta)}P_\theta,\\
P_\theta f'_iP_\theta&=\frac{\theta}{\varepsilon(\theta)}P_\theta, \qquad i=1,\ldots,4.
\end{aligned}
\end{equation}
This example is the Clifford-covariant version of the disjoint-pair construction: the visible \(Z/X\) axes are only a normal form for a mixed-Pauli pair of transverse commuting families whose phase-free cross-pairing is nondegenerate.


\bibliography{references}


\end{document}